\documentclass[]{elsarticle}

\usepackage[utf8]{inputenc}
\usepackage[T1]{fontenc}

\usepackage[english]{babel}
\usepackage{graphicx}
\usepackage{amsmath}
\usepackage{amsthm}
\usepackage{mathtools}
\usepackage{amssymb}
\usepackage{faktor}
\usepackage{tikz-cd}
\usepackage{comment}

\usepackage{todonotes}

\usepackage{hyperref}
\hypersetup{
	colorlinks,
	linkcolor={red!50!black},
	citecolor={blue!50!black},
	urlcolor={blue!80!black}
}

\usepackage[linesnumbered,onelanguage,algoruled,noend]{algorithm2e}

\newcommand{\F}{\mathbb{F}}
\newcommand{\N}{\mathbb{N}}
\renewcommand{\P}{\mathbb{P}}
\newcommand{\Q}{\mathbb{Q}}
\newcommand{\Z}{\mathbb{Z}}

\newcommand{\bfR}{\boldsymbol{\mathfrak{R}}}

\DeclareMathOperator{\CRT}{CRT}

\newcommand{\SRN}{\mathsf{SRN}}
\newcommand{\RN}{\mathsf{RN}}

\newcommand{\bE}{\boldsymbol{E}}
\newcommand{\bR}{\boldsymbol{R}}

\newcommand{\bC}{\boldsymbol{C}}
\newcommand{\bPsi}{\boldsymbol{\psi}}
\renewcommand{\bf}{\boldsymbol{f}}
\newcommand{\br}{\boldsymbol{r}}
\newcommand{\be}{\boldsymbol{e}}
\newcommand{\beps}{\Bold{\epsilon}}

\newcommand{\cB}{\mathcal{B}}
\newcommand{\C}{\mathcal{C}}
\newcommand{\D}{\mathcal{D}}
\newcommand{\cF}{\mathcal{F}}
\renewcommand{\L}{\mathcal{L}}
\newcommand{\cP}{\mathcal{P}}
\newcommand{\cS}{\mathcal{S}}

\newcommand{\ERR}{\mathcal{E}^1_\Lambda}
\newcommand{\ERRd}{\mathcal{E}^2_{\Lambda_m}}
\newcommand{\HYB}{\mathcal{H}^1_{\Li\Lu,\beps}}
\newcommand{\HYBd}{\mathcal{H}^2_{\Lmi\Lu,\beps}}
\newcommand{\HYBp}{\cB_{\Le\Lv,\bfR}^{1}}
\newcommand{\HYBpd}{\cB_{\Lme\Lv,\bfR}^{2}}

\newcommand{\SR}{S_{\bR}}
\newcommand{\SRd}{S_{\bR,2^d}}
\newcommand{\SRdb}{S_{\bR,2^d\beta}}
\newcommand{\SE}{S_{\bE}}
\newcommand{\SEh}{S_{\bE}^{h}}
\newcommand{\SEinf}{S_{\bE}^{\infty}}
\newcommand{\dm}{\bar{d}}
\newcommand{\dmi}{\bar{d}_{i}}
\newcommand{\dme}{\bar{d}_{e}}
\newcommand{\SL}{\mathbb{S}_{\Bold{\lambda}}^{\ell}}

\newcommand{\Bold}[1]{\boldsymbol{#1}}

\newcommand{\fail}{f\!ail}

\DeclareMathOperator{\cremop}{crem}
\newcommand{\crem}{\, \cremop \, }

\DeclareMathOperator{\LLL}{LLL}
\DeclareMathOperator{\Val}{val}
\DeclareMathOperator{\vr}{v}
\DeclareMathOperator{\Ev}{Ev}
\newcommand{\Evi}{\Ev^{\infty}}

\newcommand{\dist}{d}
\newcommand{\du}{\dist_u}

\newcommand{\romain}[1]{{\color{purple}#1}}

\newcommand{\commentaire}[1]{}

\theoremstyle{plain}
\newtheorem{theorem}{Theorem}[section]
\newtheorem{prop}[theorem]{Proposition}
\newtheorem{lemma}[theorem]{Lemma}

\theoremstyle{definition}
\newtheorem{exam}[theorem]{Example}
\newtheorem{defi}[theorem]{Definition}
\newtheorem{prob}[theorem]{Problem}
\newtheorem{constraint}[theorem]{Constraint}

\theoremstyle{remark}
\newtheorem{remark}[theorem]{Remark}

\newcommand{\hdu}{d_u} 
\newcommand{\hdi}{d_i} 
\newcommand{\hde}{d_e} 
\newcommand{\hdv}{d_v} 

\newcommand{\Lu}{\Lambda_{u}}
\newcommand{\Li}{\Lambda_{i}}
\newcommand{\Lv}{\Lambda_{v}}
\newcommand{\Le}{\Lambda_{e}}
\newcommand{\Lmi}{\Lambda_{m,i}}

\newcommand{\Lme}{\Lambda_{m,e}}
\newcommand{\xii}{\xi_{i}}
\newcommand{\xiu}{\xi_{u}}

\newcommand{\ie}{\textit{i.e.\ }}

\newcommand{\wrt}{\textit{w.r.t.\ }}

\newcommand{\lc}{\left\lceil}   
\newcommand{\rc}{\right\rceil}
\newcommand{\lf}{\left\lfloor}   
\newcommand{\rf}{\right\rfloor}

\allowdisplaybreaks

\journal{JSC - Journal of Symbolic Computation}

\makeatletter
\def\ps@pprintTitle{%
	\let\@oddhead\@empty
	\let\@evenhead\@empty
	\def\@oddfoot{\centerline{\thepage}}%
	\let\@evenfoot\@oddfoot}
\makeatother
\begin{document}
	
	\begin{frontmatter}

		\title{Simultaneous Rational Number Codes: \\Decoding Beyond Half the Minimum Distance \\with Multiplicities and Bad Primes}

		\author{Matteo Abbondati,
			Eleonora Guerrini,
			Romain Lebreton}

		\affiliation{organization={LIRMM - University of Montpellier},
			addressline={161, Rue Ada}, 
			city={Montpellier},
			postcode={34095}, 
			country={FRANCE}}

		\begin{abstract}
			
			In this paper, we extend the work of~\cite{abbondati2024decoding} on decoding simultaneous rational number codes by addressing two important scenarios: multiplicities and the presence of bad primes (divisors of denominators). 
			First, we generalize previous results to multiplicity rational codes by
			considering modular reductions with respect to prime power moduli. Then,
			using hybrid analysis techniques, we extend our approach to vectors of
			fractions that may present bad primes. 
			
			Our contributions include: a decoding algorithm for simultaneous
			rational number reconstruction with multiplicities, a rigorous analysis
			of the algorithm's failure probability that generalizes several previous
			results, an extension to a hybrid model handling situations where not
			all errors can be assumed random, and a unified approach to handle bad
			primes within multiplicities. The theoretical results provide a
			comprehensive probabilistic analysis of reconstruction failure in these
			more complex scenarios, advancing the state of the art in error
			correction for rational number codes.

		\end{abstract}

	\end{frontmatter}

	\section{Introduction}
	
	An efficient approach to solving linear systems in distributed computation involves reconstructing  
	a vector of fractions $\left(\frac{f_1}{g},\ldots,\frac{f_\ell}{g}\right)$, all sharing the same denominator,  
	from its modular reductions with respect to $n$ pairwise coprime elements. In this framework, a network  
	is structured around a central node, which selects a sequence of relatively prime elements  
	$\left(m_j\right)_{1 \le j \le n}$ and delegates the system solving process to the network.  
	Each node $j$ computes the solution modulo $m_j$ and transmits the reduced solution vector  
	$\left(f_i/g\ \bmod m_j\right)_{1 \le i \le \ell}$ back to the central node. The central node then  
	reconstructs the original vector through an interpolation step, formulated as a simultaneous  
	rational reconstruction problem. In the case of polynomial systems, this approach is known  
	as evaluation-interpolation~\cite{kaltofen2017early,guerrini2019polynomial}, whereas for integer  
	systems, it corresponds to modular reduction followed by reconstruction via the Chinese Remainder  
	Theorem~\cite{cabay1971exact}. In this paper, we focus on the latter case.
	
	\paragraph{Context of this paper}
	During data reconstruction, the central node may receive incorrect reductions due to computational errors,  
	faulty or untrusted nodes, or network noise. For that reason, it is of great help to look at decoding algorithms in error correcting codes. 
	Viewing the modular
	reductions as coordinates of an error correcting code  enables us to reconstruct the correct solution as long as the number
	of erroneous reductions is below a certain value, corresponding to the unique decoding radius of a
	code. In presence of more errors, there exist two
	possible approaches in coding theory to correct beyond the unique decoding radius of the code;
	either decoding algorithms which return a list of all codewords within a certain distance of the received word (list decoding) or, by interleaving techniques, obtain positive decoding results under probabilistic assumptions on random errors corrupting $\ell$ code-words on the same positions. 
	In this paper, we focus on interleaving techniques as they fit in the simultaneous reconstruction problem. 
	Note that, a decoding algorithm working under this latter approach must inevitably fail for some instances, as
	beyond unique decoding radius there can be many codewords around a given instance. Here the failure
	probability is intended as      the proportion of received words, within a       given distance from
	the codeword $\bf/g$, for which the reconstruction fails. 
	
	In this work we consider the simultaneous rational reconstruction problem with $m_j = p_j^{\lambda_j}$ for a sequence of distinct prime numbers $p_1,\ldots,p_n$ and relative multiplicities $\lambda_1,\ldots,\lambda_n>0$. 
	One advantage of considering reductions with multiplicities is that
	solving a linear system modulo $p^\lambda$ is asymptotically faster than solving
	it modulo $p_1,\dots,p_\lambda$ (see \cite{moenck_approximate_1979,dixon_exact_1982,storjohann_shifted_2005} or \cite[Chapter 3]{lebreton_contributions_2012} for a survey).
	
	The prime numbers for which the modular reductions are not defined (divisors of the denominator $g$) are referred to as \textit{bad primes}.
	To the best of our knowledge, this work represents the first study of rational number codes in a context with multiplicities and bad primes.
	
	Taking inspiration from \cite{khonji2010output}, we could define the rational number code in terms of modular reductions to a generic sequence of $n$
	coprime ideals (not necessarily of the form $(p_j^{\lambda_j})$). From a purely mathematical perspective (thanks
	to the Chinese Remainder theorem), the approach of~\cite{khonji2010output} where
	coordinates are defined via modular reductions relative to any sequence of $n$ coprime ideals, is
	equivalent to ours, where each coprime ideal is generated by the power of a prime element. The
	advantage of the approach proposed here is that it allows for greater specificity in both the
	coordinates and the description of the errors affecting them.
	
	In order to generalize our first results (Theorems~\ref{thm:main1} and~\ref{thm:main2}) on the decoding with multiplicities to a context including bad primes (Theorems~\ref{thm:main1bad primes} and~\ref{thm:main2bad primes}), we need to consider (as in~\cite{guerrini2023simultaneous} for the rational function version of the simultaneous rational reconstruction problem) a hybrid error model consisting of random evaluation errors and fixed valuation errors. 
	In Section~\ref{sec:Hybrid Decoding} we present this error model for more general sets of errors. We note that this more general version has been independently introduced in~\cite{brakensiek2025unique} for IRS codes, folded Reed-Solomon and multiplicity codes decoding, as the \textit{semi-adversarial error model}. 
	The analysis presented in this work is adapted to the polynomial system case, obtaining similar results for the simultaneous rational function reconstruction with errors, these results will be presented in a subsequent paper~\cite{abbondati2025rationalfunction}.
	Compared to said results coming from the adaptation of our analysis technique to the rational function case, in~\cite{brakensiek2025unique} the authors obtained a weaker failure probability bound (depending linearly instead of exponentially on the error size) but, for a given total amount of errors (random and fixed), the proportion of adversarial (fixed) errors they manage to correct, though asymptotically equivalent, is higher compared to our results.

	\paragraph{Previous results}
	The approach of this paper generalizes, and matches or even improves several
	previous results in different ways. In the  polynomial case, the codes
	used for the recovery of a vector of polynomials from partially erroneous
	evaluations are Interleaved Reed-Solomon codes (IRS), whose best known analysis
	of the decoding failure probability is provided
	in~\cite{schmidt2009collaborative} and then generalized to the rational function
	case in~\cite{guerrini2019polynomial}.
	
	The integer counterpart of IRS codes are to the so-called Interleaved Chinese remainder codes (ICR),
	for which a first heuristic analysis of the decoding failure probability  was provided
	in~\cite{li2013decoding} and made rigorous in~\cite{abbondati2023probabilistic}. 
	
	While there have been various studies on the rational function
	case~\cite{kaltofen2017early,zappatore2020simultaneous,guerrini2023simultaneous}, the rational
	number context had not been investigated until~\cite{abbondati2024decoding}.
	
	In any case the extensive literature addressing
	these problems both in the polynomial
	\cite{mcclellan1977exact,boyer2014numerical,guerrini2019polynomial,kaltofen2020hermite,guerrini2023simultaneous}
	and the integer \cite{cabay1971exact,lipson1971chinese,abbondati2023probabilistic} contexts rarely
	shows unified methods, and the techniques used are very specific to the case studied. In
	\cite{abbondati2024decoding} the authors analyzing both the rational functions and the rational
	numbers reconstruction problems (in absence of multiplicities and poles/bad primes), proved it is possible to recover the correct solution vector for almost all instances.
	
	\paragraph{Contributions of this paper}
	The main results presented are the following:
	\begin{itemize}
		\item A decoding approach to address the simultaneous rational number reconstruction with errors (Problem~\ref{SRNRwE_Problem}) including multiplicities, as well as the relative decoding algorithm (Algorithm~\ref{algoSRN}).
		\item A detailed analysis of the failure probability of the algorithm, that generalizes several previous results in \cite{abbondati2024decoding}: see Theorem~\ref{thm:main1} and Theorem~\ref{thm:main2}.
		\item The extension of the analysis to a hybrid model including random and non-random errors, addressing situations where not all errors can be assumed random: see Theorem~\ref{thm:main1hyb}, Theorem~\ref{thm:main2hyb}. 
		\item The merging of the hybrid model with our decoding approach, to handle bad primes within multiplicities, and relative decoding failure analysis: see Theorem~\ref{thm:main1bad primes} and Theorem~\ref{thm:main2bad primes}. 
	\end{itemize}
	
	Our methodology and our results can also be adapted to the rational function
	case. For the sake of readability, we have chosen to focus on the simultaneous
	reconstruction of rational numbers in this paper. However, it is worth reporting
	that adapting our results would improve upon the existing
	analysis of \cite{guerrini2023simultaneous} in the sense that the failure
	probability bound obtained decreases exponentially (not linearly) with respect
	to the decoding algorithm distance parameter, and the dependency  on the
	choice of the multiplicities can be removed (see multiplicity balancing in \cite[Theorem 3.4]{guerrini2023simultaneous}).

	\subsection{Notations and preliminary definitions}
	\label{sec:NotationDefinitions}
	We will denote vectors with bold letters $\bf,\br,\Bold{c},\ldots$.
	For $m\in\Z$ with $\Z/m\Z$ we will denote the quotient ring modulo the ideal $(m)$, while $[x]_m$
	will denote the modular element $x\bmod m\in\Z/m\Z$ and $\cP(m)$ will denote the set of primes
	dividing $m$.
	Given an indexed family of rings $\left\{A_j\right\}_{1 \le j \le n}$, we let $\prod_{j=1}^n A_j$ be their
	Cartesian product.
	
	Given a vector of modular reductions
	$\br\in\prod_{j=1}^n\Z/p_j^{\lambda_j}$ we use the corresponding capital letter $R$ denote its
	unique interpolant constructed via the Chinese remainder theorem modulo
	$N\coloneq\prod_{j=1}^np_j^{\lambda_j}$.
	
	We let $\Val_{p}:\Z \longrightarrow
	\N\cup\{\infty\}$ be the valuation function over $\Z$ with respect to the prime number $p$, whose output is the
	highest power of $p$ dividing the input, where we set by convention its
	value to be $\infty$ when the input is $0$.
	
	Dealing with a fixed sequence of precisions $\lambda_1,\ldots,\lambda_n$, we truncate the valuation
	function considering $\nu_{p_j}(m) \coloneq\min\{\Val_{p_j}(m),\lambda_j\}$, so that $\nu_{p_j}(a) = \nu_{p_j}(b)$ when $a = b \bmod p_j^{\lambda_j}$.

	When computing the valuation of a vector we set
	$\nu\left(\bf\right)\coloneq\min_i\{\nu\left(f_i\right)\}$. Given the sequence of
	multiplicities $\lambda_1,\ldots,\lambda_n>0$, we define the parameter
	$L\coloneq\sum_{j=1}^n\lambda_j$. For us all the vectors of fractions $\bf/g$ sharing the same
	denominator will always be reduced, i.e. they satisfy $\gcd\left(\gcd(\bf),g\right) = 1$.

	\paragraph{Simultaneous rational number reconstruction with errors (SRNRwE)}

	To quantify errors and to establish the correction capacity of the code we are
	going to use, we need a notion of distance between words. In a context with
	multiplicities where the coordinates are modular reductions relative to moduli
	specified by different precisions $\lambda_1,\ldots,\lambda_n$, it is classical
	to consider (see for example~\cite{kaltofen2020hermite,guerrini2023simultaneous}) a minimal error
	index distance in which each modular reduction is regarded as a truncated
	development, and the whole tail starting from the first error index in such
	development is considered erroneous. Furthermore, to take into account that each
	coordinate depends on a different prime number $p_j$, it is classical to use a
	weighted Hamming distance (see for example~\cite{abbondati2024decoding}), thus
	we are going to consider the following definition:
	
	\begin{defi}[Distance - Integer case]
		\label{def:IntDistance}
		Let \mbox{$\bR^1, \bR^2 \in (\prod_{j=1}^n\Z/p_j^{\lambda_j}\Z)^{\ell}$} be two
		\mbox{$\ell\times n$} matrices, where each column $\br_j^1,\br_j^2$ belongs to $\left(\Z/p_j^{\lambda_j}\right)^\ell$. We
		define their error support as  
		$\xi_{\bR^1,\bR^2} \coloneq \{ j: \br_j^1 \ne \br_j^2\}$ and their error
		locator as the product $\Lambda_{\bR^1,\bR^2}\coloneq\prod_{j \in \xi_{\bR^1, \bR^2}}
		p_j^{\lambda_j - \mu_j}$, where $\mu_j\coloneq\nu_{p_j}\left(\br_j^1  -
		\br_j^2\right)$ represents the minimal error index for the development around the
		prime $p_j$. The distance between $\bR^1$ and $\bR^2$ is defined as
		$d(\bR^1,\bR^2)\coloneq\log_2(\Lambda_{\bR^1,\bR^2})$.
	\end{defi}
	
	The problem of simultaneous rational number reconstruction with errors is then:
	
	\begin{prob}[SRNRwE]
		\label{SRNRwE_Problem}
		Given $\ell >0$, $n$ distinct primes $p_1<\ldots< p_n$ with associated multiplicities
		$\lambda_1,\ldots,\lambda_n$, a received matrix $\bR \in
		(\prod_{j=1}^n\Z/p_j^{\lambda_j}\Z)^{\ell}$, an error parameter $d$ and two bounds $F, G$ such
		that $FG<N/2$, find a reduced vector of fractions $\left(f_1/g,\ldots,
		f_\ell/g\right)\in\Q^{\ell}$ such that
		\begin{enumerate}
			\item $d\left(\bigl([f_i/g]_{p_j^{\lambda_j}}\bigr)_{i,j},\bR\right)\leq d$,
			\item for all $1 \leq i \leq \ell$, $|f_i|<F$, $0<g<G$ and $\gcd(g,N)=1$.
		\end{enumerate}
	\end{prob}
	In the above we have that $\gcd(g,N)=1$ so that the
	reductions $[f_i/g]_{p_j^{\lambda_j}}$ are well-defined. We are going to drop this hypothesis in
	Section~\ref{Sec:bad primes}, when solving a more general version of
	the SRNRwE problem, allowing for the presence of bad primes.

	This problem can be reduced to the simultaneous error correction of $\ell$ code words
	(sharing the same denominator) for the multiplicity version of rational number codes. Without multiplicities (i.e. when $N$ is square-free)
	this code is the natural rational extension of Chinese remainder codes \cite{goldreich1999chinese}, and can be referred to as rational number codes,
	extensively studied in~\cite{abbondati2024decoding}. It seems these rational codes were part of the
	folklore; to the best of our knowledge, they were formally introduced in the language of coding
	theory by Pernet in \cite[$\S$ 2.5.2]{pernet2014high}, whereas~\cite{bohm2015use}
	works with redundant residue number systems.

	The condition $FG<N/2$ guarantees an injective encoding, whose proof will be given in Proposition~\ref{prop:InjectiveEncodingbad primes}
	when introducing the multi-precision encoding (see Definition~\ref{def:MultiprecisionEncoding}) which is a
	generalization of our current encoding in presence of bad primes.
	
	A long series of papers can be found in the literature where
	evaluation-interpolation is used for linear systems solving, as
	\cite{mcclellan1977exact,villard1997study,monagan2004maximal,olesh2007vector,rosenkilde2016algorithms}.
	Our contributions in this paper concern error correction beyond guaranteed
	uniqueness. This means that the solution to the problem will not always be
	unique. In this rare case, our decoding algorithm returns a decoding failure.
	We analyze the probability of failure in detail.
	
	\medskip
	
	The paper is structured as follows: In Section~\ref{Sec:SRN} we introduce the
	simultaneous rational number codes whose decoding solves
	Problem~\ref{SRNRwE_Problem} as well as the corresponding decoding
	Algorithm~\ref{algoSRN}. We study the failure probability of our decoding
	algorithm for error parameters larger than the unique decoding radius of the
	code. We note that this analysis generalizes the results
	of~\cite{abbondati2024decoding} to the multiplicity case, it thus follows the
	same broad lines except for some technical details (see
	Lemma~\ref{lm:sumOverDivisorsINT}).
	
	In Section~\ref{sec:Hybrid Decoding}, we adapt our analysis technique to the hybrid distribution
	model of~\cite{guerrini2023simultaneous} in which not all errors are supposed to be random, but some
	of them are fixed, either because of specific error patterns introduced by malicious entities or
	because of specific faults of the network nodes.
	
	Then, in Section~\ref{Sec:bad primes}, by considering the multi-precision encoding
	of~\cite{guerrini2023simultaneous}, we apply the hybrid approach to generalize our analysis to the
	case of reductions with multiplicities and bad primes, \ie we drop the hypothesis
	$\gcd\left(g, N\right) = 1$.

	\section{Simultaneous multiplicity rational number codes}\label{Sec:SRN} We can define an error
	correcting code associated to Problem~\ref{SRNRwE_Problem}. Code words are the encoding of reduced
	vectors of rational numbers $\left(f_1/g,\ldots, f_\ell/g\right)$ sharing the same denominator and
	such that $0<g<G$, and $|f_i|<F$ for all $i=1,\ldots,\ell.$ 
	\begin{defi}\label{def:SRN} Given $n$ distinct primes $p_1,\ldots,p_n$ with relative multiplicities
		$\lambda_1,\ldots,\lambda_n$, two positive bounds $F, G$ such that $FG<N/2$ and an integer
		$\ell>0$, we define the \textit{simultaneous multiplicity rational number code} as the set of
		matrices 
		\begin{equation*}
			\SRN_{\ell}(N;F,G)\coloneq
			\left\{
			\begin{pmatrix}
				\left[\frac{f_i}{g}\right]_{p_j^{\lambda_j}}
			\end{pmatrix}_{\substack{1 \leq i \leq \ell \\ 1 \leq j \leq n}}
			:
			\begin{array}{c}
				|f_i|<F, \quad 0<g<G,\\
				\gcd(f_1,\ldots,f_\ell,g)=1\\
				\gcd(N,g)=1
			\end{array}
			\right\}.
		\end{equation*}
		
		We will refer to SRN codes for short if parameters are not relevant.
	\end{defi}

	Note that when $G=2$ and $\ell=1$, we obtain the
	interleaving of $\RN_{\ell}(N;F,G)$ codes and if  $\ell>1$, then $\RN_{\ell}(N;F,G)$  gives the rational codes with multiplicity as described in \cite[$\S$2.5.2]{pernet2014high}.  When dealing with rationals numbers,  the denomination \emph{simultaneous} comes from the rational function case, where it is related to 
	simultaneous rational function reconstruction, \ie the variant of Problem~\ref{SRNRwE_Problem}
	for rational functions without errors~\cite{olesh2007vector, rosenkilde2016algorithms, guerrini_uniqueness_2020}.
	
	In the next section, we will see that the common denominator property is necessary to be able to take advantage in the 
	key equations of the fact that the $\ell$ RN codewords share the same error supports. 

	The condition $\gcd(f_1,\ldots,f_\ell,g)=1$, which is going to be used in the proof of
	Lemma~\ref{BaseLemmaRN}, reflects that the solution vector we seek to reconstruct is a reduced
	vector of rational numbers.
	\begin{remark}
		A bounded distance decoding algorithm for the above code which is able to correct errors up to a
		distance $d$, can be used to solve Problem~\ref{SRNRwE_Problem} with error parameter $d$.
	\end{remark}

	\subsection{Minimal distance}
	
	The distance $d(\C)\coloneq\min_{c_1 \neq c_2\in \C} d(c_1,c_2)$ of a code $\C$ plays an important
	role in coding theory to assess the amount of data one can correct. A classic result states that one can correct up to half of the Minimal distance and that there is no guarantee on the decoding success beyond this quantity.
	\begin{theorem}\label{thm:minDisRN} Let $N, F, G$ as in Definition~\ref{def:SRN}. The distance of an
		RN code satisfies $d(\RN(N;F,G)) > \log\left( \frac{N}{2FG}  \right) $.
	\end{theorem}
	This  result has the advantage of being independent of the moduli $p_j$. However, the gap between
	$d(\RN(N;F,G))$ and $\log\left(N/(2FG)\right)$ depends on the moduli. 
	Even so, there exists a family of RN codes such that $d(\RN(N;F,G)) \le
	\log\left(N/((F-1)(G-1))\right)$, \ie{} the gap is small \cite[$\S$2.5.2]{pernet2014high}.
	We can generalize Theorem~\ref{thm:minDisRN} to SRN codes with multiplicities %
	as follows:
	\begin{lemma}
		\label{lm:minDistSRNcode}
		We have $d(\SRN_{\ell}(N;F,G)) > \log \left( \frac{N}{2FG}  \right) $.
	\end{lemma}
	\begin{proof}
		Let  $\bC_1= \left([f_i/g]_{p_j^{\lambda_j}}\right)_{i,j}$ and $\bC_2=
		\left([f_i'/g']_{p_j^{\lambda_j}}\right)_{i,j}$ be two code words. Setting
		$Y\coloneq\prod_{j\notin \xi_{\bC_1,\bC_2}} p_j^{\mu_j}$, with $\mu_j =
		\nu_{p_j}\left([\bf/g]_{p_j^{\lambda_j}} - [\bf' /g']_{p_j^{\lambda_j}}\right)$. Since
		$\gcd\left(Y,g\right) = \gcd\left(Y,g'\right)= 1$, we have that  
		$Y| (\bf g' - \bf'  g)$. Since  $\|\bf\|_{\infty}, \|\bf' \|_{\infty}<F$, and
		$0<g,g'<G$ we have $Y<2FG$. Using the relation $Y=N/ \Lambda_{\bC_1,\bC_2}$,
		we bound $d(\bC_1,\bC_2) = \log(\Lambda_{\bC_1,\bC_2}) = \log(N/Y) > \log(N/2FG)$.
	\end{proof}
	
	\subsection{Unique decoding}
	
	A unique decoding function $D$ of capacity $d$ is a function from the ambient space to the code such
	that $D(r)=c$ for all code word $c$ and all $r$ such that $d(r,c) \le t$.
	For codes equipped with the Hamming distance, there
	exists such a decoding function of capacity $d$ if and only if $2d < d(\C)$. Pernet gives a polynomial time unique decoding algorithm for RN codes of capacity
	$\log(\sqrt{N/(2FG)}) = (1/2) \log(N/(2FG))$ for the weighted Hamming distance \cite[Corollary
	2.5.2]{pernet2014high}. Note that if no such
	decoding function exists, then no decoding algorithm can exist.
	
	For SRN codes equipped with the weighted Hamming distance, the result is slightly different. If $2d
	< d(\C)$, then there exists such a decoding function of capacity $d$. However, the converse is false
	in the strict sense of the term. Indeed, whereas proving that there can not exist a decoding function
	when $2d = d(\C)$, one takes $c_1,c_2 \in \C$ such that $d(\C)=d(c_1,c_2)$, and constructs $r$ as
	the middle of $c_1$ and $c_2$, \ie{} with $d(c_1,r)=d(c_2,r)=d(c_1,c_2)/2$. Thanks to that, we obtain the
	contradiction that a decoding function would have to map $r$ to both $c_1$ and $c_2$. However, it is
	impossible to construct $r$ as the middle of $c_1$ and $c_2$ with the weighted Hamming distance
	associated to distinct primes. Still, the essence of the result remains correct, and if $2d = d(\C)
	+ \varepsilon$ for a small $\varepsilon$, then we can construct $r$ such that $d(c_1,r),d(c_2,r) \le
	(d(c_1,c_2) + \varepsilon)/2=d$, and no decoding function of capacity $d$ can exist.

	Thanks to Lemma~\ref{lm:minDistSRNcode}, we know that a unique decoding function of capacity $d$ for
	$SRN$ codes can exist only if $d<\log\left(\sqrt{N/2FG}\right)$ (see Proposition~\ref{Unicity} for a
	proof in the case of bad primes).
	
	One workaround in coding theory, when no unique decoding function can exist, consists of having
	decoding functions which can output "decoding failure" when the code word within the decoding
	capacity is not unique.
	
	The aim of the paper is to properly analyze the decoding failure probability of a decoding algorithm
	for SRN codes beyond the uniqueness capacity.  
	It is worth of note that our decoding algorithm (Algorithm~\ref{algoSRN}), despite being aimed at correcting errors beyond unique decoding, outputs the unique decoding solution whenever $d<\log\left(\sqrt{N/2FG}\right)$ (see Remark~\ref{rmk:unicity}).
	
	\subsection{Decoding SRN codes}
	
	This section presents our first contribution: a decoder of SRN codes of capacity beyond
	$\frac{d(C)}{2}$. This decoder, is a slight modification of the decoder presented
	in~\cite{abbondati2024decoding} for SRN codes without multiplicities, and it is based on the
	interleaved Chinese remainder (ICR) codes decoder of
	\cite{li2013decoding,abbondati2023probabilistic}, which are a special case of SRN when $g=1$ and $N$
	is square-free.
	Let $\bR \coloneq (r_{i,j})_{\substack{1 \leq i \leq \ell \\ 1 \leq j \leq n}}$ be the received
	matrix.
	
	For any code word $\bC\in \SRN_{\ell}(N;F,G)$, we can write $\bR=\bC+\bE$ for some error
	matrix $\bE$ (which depends on $\bR$ and $\bC$).
	Thanks to the Chinese remainder theorem, we can view each row of the matrix as modular elements in
	$\Z/N\Z$, and the ambient space for the code can be viewed as $\Z_N^\ell$, thus for every  $1\le
	i\le \ell$ we can write $R_i=C_i+E_i$ with $C_i=[f_i/g]_N$ for some $f_i,g$.
	
	Letting $\Lambda\coloneq\Lambda_{\bC,\bR} = \prod_{j\in\xi_{\bC,\bR}} p _j^{\lambda_j - \mu_j}$,
	with $\mu_j = \nu_{p_j}([\bf/g]_{p_j^{\lambda_j}} - \br_j)$ we conclude that the
	system of $\ell$ equations holds: 
	\begin{equation}
		\label{sysEQ1}
		\Lambda f_i  = \Lambda g R_i \bmod N \text{ for } i=1,\ldots,\ell
	\end{equation}
	with unknowns $\Lambda,g,f_1,\ldots,f_\ell$. 
	
	We linearize it thanks to the substitution $\varphi
	\leftarrow \Lambda g$ and $\psi_i \leftarrow \Lambda f_i$; the resulting equations 
	\begin{equation}
		\psi_i  = \varphi R_i \bmod N \text{ for } i=1,\ldots,\ell
	\end{equation}
	are called the \emph{key equations}. The solutions $(\varphi,\psi_1,\ldots,\psi_\ell)$ are vectors
	in the lattice $\L\subseteq\Z^{\ell + 1}$ spanned by the rows of the integer matrix
	\begin{equation}
		\label{eq:lattice}
		\L=\operatorname{Span}\left(\begin{matrix}
			1 & R_1 & \cdots & R_\ell \\
			0 & N & \cdots & 0 \\
			\vdots & \vdots & \ddots & \vdots \\
			0 & 0 & \cdots & N
		\end{matrix}\right).
	\end{equation}
	In particular if $\Lambda\leq 2^d$ for some distance parameter $d$, the solution vector $v_C
	\coloneq \left(\Lambda g,\Lambda f_1,\ldots,\Lambda f_\ell\right)$ belongs to the set
	\begin{equation*}
		S_{\bR,2^d}\coloneq\{(\varphi,\psi_1,\ldots,\psi_\ell)\in\L:0<\varphi<2^d G, |\psi_i|<2^d F\}.
	\end{equation*}
	Note that the condition $\Lambda_{\bC,\bR}\leq 2^d$ means that $\bC$ is close to $\bR$ for the
	weighted Hamming distance.

	The decoding strategy consists in compute an element of $S_{\bR,2^d}$ and try to recover $v_{\bC}$ by dividing all the entries by the first one in order to obtain $\left(f_1/g,\ldots,f_\ell/g\right)$.
	There are two main
	aspects inherent to this procedure. The first one is algorithmic, and it is relative to a choice of
	how to compute  an element in $S_{\bR,2^d}$, the second one is probabilistic, and it is relative to
	the estimation of the probability that this element is a multiple of the solution vector $v_{\bC}$.
	Concerning the analysis of this second aspect, more will be said in Section~\ref{Sec.AnalysisRN}.
	For the moment we wish to describe the algorithmic aspect at a high level of generality. For this we
	will assume to have at our disposal an algorithm $\mathcal{ASVP}_{\infty}$ which solves the
	following problem:
	\begin{prob}[$\text{SVP}^{\beta}_{\|\cdot\|_{\infty}}$]
		\label{Prob_beta_approx}
		Given a basis $\{v_0,\ldots,v_{\ell}\}$ of a lattice $\L$ and an approximation constant
		$\beta\geq 1$, find a non-zero vector $w\in\L$ such that $\|w\|_{\infty}\leq
		\beta\lambda_\infty(\L)$,  
		where $\lambda_\infty(\L)$ is the minimum $\|\cdot\|_{\infty}$-norm of the non-zero vectors in
		$\L$.
	\end{prob}
	We refer the reader to \cite{aggarwal2018improved} for state-of-the-art algorithms solving
	Problem~\ref{Prob_beta_approx}. Without loss of generality, we will assume that the output $w$ of
	the algorithm $\mathcal{ASVP}_{\infty}$ satisfies $w_0 \ge 0$ (both $\pm w$ are short vectors). We
	will also assume that $w$ is $\L$-reduced: 
	\begin{defi}
		\label{DefiL-reduced}
		Given a lattice $\L$, a vector $v\in\L$ is said to be $\L-$reduced if, for $c \in
		\Z\setminus\{0\}$, $(1/c) \cdot v\in\L\Rightarrow c=\pm 1.$
	\end{defi}
	Because the size constraints in $S_{\bR,2^d}$ do not correspond exactly to conditions on the
	$\|\cdot\|_\infty$ norm, we need to introduce a scaling operator $\sigma_{F,G} : \Q^{\ell+1}
	\rightarrow \Q^{\ell+1}$ such that $\sigma_{F,G}((v_0,v_1,\dots,v_\ell))\coloneq(v_0 F,v_1
	G,\dots,v_\ell G)$.
	This scaling will transform $\L$ into the scaled lattice $\bar{\L} \coloneq \sigma_{F,G}(\L)$, and
	our solution set $S_{\bR,2^d}$ into
	\begin{equation*}
		S'_{\bR,2^d}\coloneq\sigma_{F,G}(S_{\bR,2^d})=\{(\varphi,\psi_1,\ldots,\psi_\ell)\in\bar{\L}:0<\varphi<2^d FG, |\psi_i|<2^d FG\}.
	\end{equation*}
	Therefore, a vector $v'\in\bar{\L}$ which satisfies $\|v'\|_\infty < 2^d FG$ must belong to
	$S'_{\bR,2^d}$.
	A candidate solution $v_s$ can be obtained by computing a scaled short vector $\bar{v}_s \coloneq
	\mathcal{ASVP}_{\infty}(\bar{\L})$, and unscaling it $v_s \coloneq \sigma_{F,G}^{-1}(\bar{v}_s)$.

	We can now prove that, provided that $\bR$ is relatively close to the code (see Constraint~\ref{c_1}
	below), since $v_s$ is a $\beta$-approximation of the shortest vector, it belongs to a slightly
	larger solution set.

	\begin{constraint}\label{c_1} There exists a code word $\bC$ such that $\Lambda_{\bC,\bR} \leq 2^d$.
		
	\end{constraint}
	
	\begin{lemma}
		\label{Lemma1}
		Assuming Constraint~\ref{c_1}, we have that $v_s \in S_{\bR}\coloneq S_{\bR,2^{d}\beta}$.
	\end{lemma}
	\begin{proof}
		We know that
		$\|\bar{v}_s\|_\infty\leq\beta\lambda_\infty(\bar{\L})\leq\beta\|\sigma_{F,G}(v_C)\|_\infty <
		\beta\Lambda FG \leq \beta 2^d FG$.
		
		Since we assumed that $(\bar{v}_{s})_0 \ge 0$, we have $\bar{v}_s\in S'_{\bR,2^{d}\beta}$ and
		$v_s\in S_{\bR,2^{d}\beta}$.
	\end{proof}
	We notice that assuming Constraint~\ref{c_1} we also have $v_{\bC} \in S_{\bR}$. 
	Following the error model of SRN codes, one could independently decode each row, which corresponds to an RN code, but the information that the errors share the same support would not be exploited. So, instead, we perform so-called collaborative decoding of $\ell$ RN codeword together, that is a SRN code word, to take advantage of this common support.
	We can now state
	our decoding algorithm for SRN codes.
	\begin{algorithm}
		\caption{SRN codes decoder.}
		\label{algoSRN}
		\SetAlgoLined \SetKw{KwBy}{par} \KwIn{$\SRN_\ell(N;F,G)$, received word $\bR$, distance
			bound $d$} \KwOut{A code word $\bC$ s.t. $d(\bC,\bR) \leq d$ or ``decoding failure''}
		
		\vspace{5pt}
		Let $\bar{\L} \coloneq \sigma_{F,G}(\L)$ be the scaled lattice of $\L$ defined in
		Equation~\eqref{eq:lattice} \\
		Compute a short vector  $\bar{v}_s \coloneq \mathcal{ASVP}_{\infty} (\bar{\L})$ \\
		Unscale the vector: $v_s = (\varphi,\psi_{1},\dots,\psi_{\ell}) \coloneq
		\sigma_{F,G}^{-1}(\bar{v}_s)$ \\
		Let $\eta \coloneq \gcd(\varphi,\psi_{1},\dots,\psi_{\ell})$, $\varphi'\coloneq\varphi/\eta$ and
		$\forall j, \ \psi'_j\coloneq\psi_j/\eta$ \label{step:lambda}\\
		\If{$\eta \le 2^d$, $\gcd(\varphi',N)=1$, $|\varphi'| < G$ and $\forall j, \ |\psi_j'|< F$}
		{\textbf{return} $(C_1,\dots,C_\ell)\coloneq(\psi_1'/\varphi', \dots, \psi_\ell'/\varphi')$}
		\lElse{ \textbf{return} "decoding failure"}
	\end{algorithm}
	
	\commentaire{ \romain{ Remark : We don't need the condition $\varphi \neq 0$ in the algorithm.
			Indeed, we know that $v_s \neq (0,\dots,0)$. If $\varphi = 0$, then $\psi_i = 0 \bmod N$ for all
			$i$, so $\psi_i = 0$ for all $i$ because of size constraints on $\psi_i$. } }
	
	\subsection[A particular sub-routine: LLL]{A particular sub-routine: $\LLL$} 
	\label{subsec:LLL} 
	We remark that the complexity of
	Algorithm~\ref{algoSRN} is mainly determined by the complexity of the sub-routine
	$\mathcal{ASVP}_{\infty}$. In particular the authors of \cite{aggarwal2018improved} showed that the
	space and time complexity for the resolution of Problem~\ref{Prob_beta_approx} are significantly
	larger than the relative costs for the resolution of the $\ell_2-$norm version of the same problem.
	\begin{prob}[$\text{SVP}^{\gamma}_{\|\cdot\|_{2}}$]
		\label{Prob_gamma_approx_2norm}
		Given a basis $\{v_0,\ldots,v_{\ell}\}$ of a lattice $\L$ and an approximation constant
		$\gamma\geq 1$, find a non-zero vector $w\in\L$ such that $\|w\|_{2}\leq \gamma\lambda_2(\L)$,  
		where $\lambda_2(\L)$ is the minimum $\|\cdot\|_{2}$-norm of the non-zero vectors in $\L$.
	\end{prob}

	\begin{remark}
		A $\gamma$-approximation SVP for the $\ell_2-$norm yields a $\gamma\sqrt{\ell +1}$-approximation
		SVP for the $\ell_\infty-$norm : If $w=\mathcal{ASVP}_{2}(\L)$ and $s_2$ (resp. $s_\infty$) is
		one of the shortest vector for the $\ell_2-$norm (resp. $\ell_\infty-$norm), then
		$\|w\|_{\infty} \leq \|w\|_{2} \leq \gamma \|s_2\|_{2} \leq \gamma \|s_\infty\|_{2} \leq \gamma
		\sqrt{\ell+1} \|s_\infty\|_{\infty} $.
	\end{remark} 
	A well known example of algorithm solving Problem~\ref{Prob_gamma_approx_2norm} is given by $\LLL$
	\cite{lenstra1982factoring}, which runs in polynomial time for the approximation factor $\gamma =
	\sqrt{2}^{\ell}$ (our lattice has dimension $\ell +1$). As Algorithm~\ref{algoSRN} does not use
	$\LLL$ as subroutine to compute a short vector $\bar{v}_s$, we are going to assume that the
	approximation constant $\beta$ satisfies the following constraint:
	\begin{constraint}
		\label{cst:beta}
		The approximation constant $\beta$ satisfies: $\beta < 3^\ell$.
	\end{constraint}
	Thanks to the above remark, Constraint~\ref{cst:beta} is automatically satisfied if using $\LLL$ as
	subroutine, it is enough to notice that $\beta =\gamma\sqrt{\ell + 1} = \sqrt{2}^\ell \sqrt{\ell +
		1}\leq 3^\ell$.
	
	The most efficient $\text{SVP}^{\gamma}_{\|\cdot\|_{2}}$ solver is given by the BKZ algorithm
	\cite{schnorr1987hierarchy}. It finds a solution of Problem~\ref{Prob_gamma_approx_2norm} with
	$\gamma = (1+\epsilon)^{\ell +1}$ in polynomial time of degree increasing as $\epsilon\rightarrow
	0$.    
	
	Furthermore, since the output of $\LLL$ or BKZ is always the first vector of a basis of the lattice,
	the following Lemma will ensure that it is $\L-$reduced.
	\begin{lemma}
		\label{lm:Lreducedbasis}
		Let $\{b_1,\ldots,b_n\}$ be a basis of a lattice $\L$, then every vector $b_i$ is $\L-$reduced.
	\end{lemma}
	\begin{proof}
		If $\frac{1}{c}b_i\in\L$ for some $c\in\Z\setminus\{0\}$, then  we can write
		$\frac{1}{c}b_i=\sum_{j=1}^n c_j b_j$ for some $c_j\in\Z$. Thus, $b_i=\sum_{j=1}^n cc_j b_j$,
		which means that $cc_i=1$, so $c=\pm 1$.
	\end{proof}

	\subsection{Correctness of Algorithm~\ref{algoSRN}}\label{Sec:Analysis}

	In this section, we study the correctness of Algorithm~\ref{algoSRN}. We start with
	Lemma~\ref{lm:succeedsImpliesCorrect} which states that the algorithm is correct when it does not
	fail.
	
	\begin{lemma}
		\label{lm:succeedsImpliesCorrect} 
		If Algorithm~\ref{algoSRN} returns $\bC$ on input
		$\bR$ and parameter $d$, then $\bC$ is a code word of $\SRN(N;F,G)$ such that
		$\dist(\bC,\bR) \le d$.
	\end{lemma}
	
	\begin{proof}
		The output vector $\bC= (\psi_1'/\varphi', \dots, \psi_\ell'/\varphi')$ is a code word of
		$\SRN(N;F,G)$ since the algorithm has verified the size conditions $|\varphi'| < G$,
		$|\psi_j'|< F$ for all $j$, and that $\gcd(\varphi',N)=1$.
		Now, we use that $(\varphi,\psi_1,\ldots,\psi_\ell)=(\eta \varphi',\eta\psi_1',\dots,\eta
		\psi_\ell')$ is in the lattice $\L$, so that $\eta(\varphi' R_i - \psi_i') = 0 \bmod N$ for all
		$i$. Dividing by the invertible $\varphi'$ modulo $N$, we obtain $\eta(R_i - C_i) = 0 \bmod N$
		for all $i$, which implies that $\nu_{p_j}(\eta)\geq \lambda_j - \mu_j =
		\nu_{p_j}(\Lambda_{\bC,\bR})$. Thus, $\Lambda_{\bC,\bR}|\eta\leq 2^d$, and we can conclude that
		$d(\bC,\bR) = \log \Lambda_{\bC,\bR} \le \log \eta \le d$.
	\end{proof}
	
	Next lemma shows that, when the algorithm fails, the short vector $v_s$ computed by sub-routine
	$\mathcal{ASVP}_{\infty}$ is not collinear to $v_{\bC}$.

	\begin{lemma}
		\label{lm:algofailurecondition}
		Assuming Constraint~\ref{c_1}, if Algorithm~\ref{algoSRN} fails, then $v_s\notin v_{\bC} \Z$.
	\end{lemma}
	\begin{proof}
		By contraposition, let's prove that if $v_s= r v_{\bC}$ for some $r\in\Z$,
		then the algorithm must succeed.
		We know that $v_s = r v_{\bC}$ is $\L$-reduced therefore $v_{\bC} = \pm v_s$ and $\eta =
		\Lambda\le 2^{d}$ using Constraint~\ref{c_1} (see Algorithm~\ref{algoSRN},
		Step~\ref{step:lambda} for $\eta$), $\varphi'=\pm g,\psi'_j=\pm f_j$ for every $j$, thus the
		algorithm succeeds.
	\end{proof}
	
	\begin{remark}
		\label{rem:failevenifbetais1}
		We emphasize here that the failure of Algorithm~\ref{algoSRN} is due to the size of the distance
		parameter $d$ (when larger than the unique decoding capacity), and not to the approximation
		factor coming from the subroutine $\mathcal{ASVP_\infty}$. When $d>\log(\sqrt{N/(2FG)})$ the
		algorithm might sometimes fail even if $\beta = 1$.
	\end{remark}
	
	The rest of this section is dedicated to the analysis of the decoding failure of
	Algorithm~\ref{algoSRN}. We will show that if $\bR$ is $\bC$ plus a random error of weighted Hamming
	distance up to approximately $\ell/(\ell+1) \log(N/(2FG))$ (see
	Section~\ref{subsect:ErrModelRatNumb} for precise error models), then this decoder is able to decode
	most of the time (see Section~\ref{SubSect.Results} for the statement of the theorem).

	\subsection{Error models}
	\label{subsect:ErrModelRatNumb}
	
	Algorithm~\ref{algoSRN} must fail on some instances when the distance parameter $d$ exceeds the
	maximum distance for which the uniqueness of the solution of Problem~\ref{SRNRwE_Problem} is
	guaranteed.
	
	We analyze the failure probability of the algorithm under two different classical error models in
	Coding Theory, already considered in previous papers
	\cite{schmidt2009collaborative,abbondati2023probabilistic,abbondati2024decoding}, specifying two
	possible distributions of the random received word $\bR$.

	\paragraph{Error Model 1}
	In this error model, we fix an error locator $\Lambda$ among the divisors of
	$N$, then we let $\ERR$ be the set of error matrices $\bE$ whose columns satisfy:
	\begin{enumerate}
		\item $\be_j = \boldsymbol{0} \text{ for all } j \text{ such that } p_j\not\in\cP(\Lambda)$,
		\item $\nu_{p_j}(\be_j) = \lambda_j - \nu_{p_j}(\Lambda) \text{ for all } j \text{ such that } p_j\in\cP(\Lambda)$.
	\end{enumerate}
	For any given code word $\boldsymbol{C}$ and error locator $\Lambda$, the
	distribution $\D^{\ERR}_{\bC}$ of random received words $\bR$ around the central
	code word $\bC$ is defined as $\bR = \bC + \bE$ for $\bE$ uniformly distributed
	in $\ERR$.
	
	We will need another point of view on the random error matrices $\bE$. For
	$i\in\{1,\dots,\ell\}$, we denote $E_i \in \Z/N\Z$ the CRT interpolant of the
	$i$-th row of $\bE$. By definition of the error valuation $\mu_j$, letting $Y
	\coloneq N/\Lambda = \prod_{j=1}^n p_j^{\mu_j}$, we have that $Y | E_i$ for
	every index $i= 1,\ldots,\ell$. We define the modular integers $E_i'\coloneq
	E_i/Y \in \Z/\Lambda\Z$.
	
	Since $\mu_j = \nu_{p_j}(\bE) = \min_i\{\nu_{p_j}(E_i)\}$, we see that $Y =
	\gcd(E_1,\ldots,E_{\ell},N)$, and that the random vector $(E_i')_{1\le i \le \ell}$ is uniformly
	distributed in the sample space 
	\begin{equation}
		\Omega_\Lambda \coloneq
		\{ (F_i)_{1\le i \le \ell} \in (\Z/\Lambda\Z)^\ell : \gcd(F_1,\dots,F_\ell,\Lambda) = 1\}.
	\end{equation}
	As we will need a more general version of $\Omega_{\Lambda}$ (for example in the proof of
	Lemma~\ref{lm:probagpe}), we state the following:
	\begin{lemma}
		\label{lm:EulerFormula}
		Given $\Lambda\in\Z$ and $\eta = \prod_{p\in\cP(\Lambda)}p_j^{\eta_j}$ be a divisor of $\Lambda$,
		then letting $\bar{\Omega}_{\Lambda,\eta} := \left\{ (F_i)_{1\le i \le \ell} \in (\Z/\Lambda\Z)^\ell :
		\gcd(F_1,\dots,F_\ell,\Lambda) = \eta \right\} $ we have
		\begin{equation*}
			\#\bar{\Omega}_{\Lambda,\eta}
			=
			\left(\frac{\Lambda}{\eta}\right)^{\ell}\prod_{p\in\cP(\Lambda/\eta)}\left(1 - \frac{1}{p^{\ell}}\right)
		\end{equation*}
	\end{lemma}
	\begin{proof}
		Thanks to the Chinese Remainder Theorem, for every $i = 1,\ldots,\ell$, we can factor each of
		the $\ell$ copies of the quotient space $\Z/\Lambda\Z$ with respect to the factors of $\Lambda$,
		and obtain that $\bar{\Omega}_{\Lambda,\eta}$ has the same cardinality as
		\begin{equation*}
			\left\{
			(\Bold{\varphi}_j)
			\in
			\prod_{p_j\in\cP(\Lambda)}\left(\Z/p_j^{\nu_{p_j}\left(\Lambda\right)}\Z\right)^{\ell}: \nu_{p_j}(\Bold{\varphi}_j)=\eta_j
			\right\}.
		\end{equation*}
		By counting the $p_j$-adic vectorial expansion coefficients of $\Bold{\varphi}_j$, we can
		compute the cardinality of the above set as
		\begin{equation*}
			\prod_{\eta_j<\nu_{p_j}(\Lambda)}
			p_j^{\ell(\nu_{p_j}(\Lambda) - \eta_j - 1)}
			\left(
			p_j^{\ell} - 1
			\right)
			=
			\left(\frac{\Lambda}{\eta}\right)^{\ell}\prod_{p\in\cP(\Lambda/\eta)}\left(1 - \frac{1}{p^{\ell}}\right). \qedhere
		\end{equation*}
	\end{proof}
	
	\paragraph{Error Model 2}
	In this error model we fix a maximal error locator $\Lambda_m$ among the divisors of $N$, then we let $\ERRd$ be the set of error matrices $\bE$ whose columns satisfy:
	\begin{enumerate}
		\item $\be_j = \boldsymbol{0} \text{ for all } j \text{ such that } p_j\not\in\cP(\Lambda_m)$,
		\item $\nu_{p_j}(\be_j) \ge \lambda_j - \nu_{p_j}(\Lambda_m) \text{ for all } j \text{ such that } p_j\in\cP(\Lambda_m)$.
	\end{enumerate}
	
	We notice that in the error model $\ERRd$, the actual error locator $\Lambda$ could be a divisor of
	$\Lambda_m$. For a code word $\boldsymbol{C}$ and a maximal error locator $\Lambda_m$, the
	distribution $\D^{\ERRd}_{\bC}$ of random received words $\bR$ around the central
	code word $\bC$ is defined as $\bR = \bC + \bE$ for $\bE$ uniformly distributed
	in $\ERRd$.
	
	\subsection{Our Results}
	\label{SubSect.Results}
	In this section we present our contributions to the analysis of the decoding failure depending on
	the given parameters. The error models previously defined will play a role in the latter but not in
	the choice of parameters. We  define a common framework for the algorithm parameters, while in 
	Subsection~\ref{Sec.AnalysisRN} we will adapt the analysis of the failure probability to the two error
	models specified above.
	In what follows we set
	\begin{equation}
		\label{dmax}
		\dm\coloneq \frac{\ell}{\ell+1} \left[\log\left(\frac{N}{2FG}\right) - \log(3\beta) \right].
	\end{equation} 
	
	\begin{remark}
		Our setting allows decoding up to a distance $d\le \dm$ that, for $\ell > 1$, can be greater
		than our estimation $\log \Bigl( \sqrt{\frac{N}{2FG}} \Bigr)$ of the unique decoding capability
		of $\SRN_{\ell}(N;F,G)$ codes.
	\end{remark}
	When fixing the decoding bound $d$ close to  $\dm$, we are likely to correct beyond the unique
	decoding radius, so we must deal with  decoding failure for some received word. Note that this
	remains valid even if $\mathcal{ASVP}_{\infty}(\bar{\L})$ gives us the exact short vector (i.e.
	$\beta=1$).
	
	Here is our first result (whose proof will be given at the end of Subsection~\ref{RN_ERR1}) relative
	to the failure probability of the decoding algorithm with respect to the error model $\ERR$.

	\begin{theorem}
		\label{thm:main1}
		Decoding Algorithm~\ref{algoSRN} on input distance parameter $d\le \bar{d}$ and a random received
		word $\bR$ uniformly distributed in $D^{\ERR}_{\bC}$, for some code word $\bC \in \SRN_\ell(N;F,G)$
		and error locator $\Lambda$ such that $\log \Lambda \leq d$, outputs the center  code word $\bC$
		of the distribution $\D^{\ERR}_{\bC}$, with a probability of failure 
		\begin{equation*}
			\P_{\fail}
			\leq 
			2^{-(\ell+1)(\dm - d)}\prod_{p\in\cP(\Lambda)}\left(\frac{1 - 1/p^{\ell + \nu_p(\Lambda)}}{1 - 1/p^{\ell}}\right).
		\end{equation*} 
	\end{theorem}

	Here is our second result (whose proof will be given at the end of Subsection~\ref{RN_ERR2})
	relative to the failure probability with respect to the error model $\ERRd$.
	\begin{theorem}
		\label{thm:main2}
		Decoding Algorithm~\ref{algoSRN} on input distance parameter $d\le \bar{d}$ and a random received
		word $\bR$ uniformly distributed in $\D^{\ERRd}_{\bC}$, for some code word $\bC \in \SRN_\ell(N;F,G)$
		and maximal error locator $\Lambda_m$ such that $\log \Lambda_m \leq d$, outputs the center
		code word $\bC$ of the distribution $\D^{\ERRd}_{\bC}$, with a probability of failure 
		\begin{equation*}
			\P_{\fail}
			\leq 
			2^{-(\ell+1)(\dm - d)}
			\prod_{p\in\cP(\Lambda_m)}
			\left(\frac{1 - 1 / p^{\ell + \nu_p(\Lambda_m)}}{1 - 1 / p^{\ell + 1}}\right).
		\end{equation*} 
	\end{theorem}
	This failure probability bound improves the one of decoding interleaved Chinese remainder codes
	$\P_{\fail}\leq 2^{-(\ell+1)(\dm - d)} + (\exp(n/p_1^{\ell-1})-1)$    
	which was only available in the special case of non-negative ($0 \le f_i$) integer code words
	($G=2$) without multiplicities ($\lambda_j =1$) \cite[Theorem 3.5]{abbondati2023probabilistic}. We
	remark that both results reduce to~\cite[Theorem 17 and 18]{abbondati2024decoding} respectively,
	when there are no multiplicities in the modular reductions of the code, \ie\ when $N$ is
	square-free.
	
	We note that in both theorems the product over the primes dividing the error locator is close to
	one; indeed we can prove the following lemma.
	
	\begin{lemma}
		Assuming that $p_1 = \min_i\{p_i\}$, given $\eta|N$ divisor of $N$ and $f(\ell)$ any function of
		the parameter $\ell>0$, we have that 
		\begin{equation*}
			\prod_{p\in \cP(\eta)}
			\left(
			\frac{1 - 1 / p^{\ell + \nu_{p}(\eta)}}{1 - 1 /p^{f(\ell)}}
			\right)
			\leq
			\frac{1}{1 - n / p_1^{f(\ell)}}.
		\end{equation*}
	\end{lemma}
	\begin{proof}
		We start noticing that for each factor in the product we have
		\begin{equation*}
			\frac{1 - 1 / p^{\ell + \nu_{p}(\eta)}}{1 - 1 /p^{f(\ell)}}
			\leq
			\frac{1}{1 - 1 / p_1^{f(\ell)}}
		\end{equation*}
		Furthermore $\prod_{p\in \cP(\eta)}(1 - 1/p^{f(\ell)}) \geq (1 - 1/p_1^{f(\ell)})^n \geq 1 - n/p_1^{f(\ell)}$, from which the statement follows.
	\end{proof}

	\begin{remark}
		We give a scenario which highlights how Theorem~\ref{thm:main2} can be used in practice. Assume
		that a code is fixed such that $\log(N/(6FG\beta)) = 20$, so that with an interleaving parameter
		$\ell = 4$, one has $\dm = 16$. If one wishes to ensure that the failure probability is less
		than a target probability of $2^{-30}$, then Theorem~\ref{thm:main2} states that choosing the
		distance parameter of the decoder $d = 10$, ensures that for any random error uniformly
		distributed on a maximal error locator $\Lambda_m$ such that $\log \Lambda_m \le d$, the failure
		probability is less than $2^{-30}$.	    
	\end{remark}

	\subsection{Analysis of the decoding failure probability}
	\label{Sec.AnalysisRN}
	
	For any $\bR$ uniformly distributed in $\D^{\ERR}_{\bC}$ (as in Theorem~\ref{thm:main1}), Constraint~\ref{c_1} is
	satisfied. Thus, thanks to Lemma~\ref{Lemma1}, we can assume that $v_s\in  S_{\bR}=
	S_{\bR,2^{d}\beta}$.
	
	\subsubsection{Decoding failure probability with respect to the first error model} 
	\label{RN_ERR1}
	If Algorithm~\ref{algoSRN} fails, then $v_s\notin v_{\bC}\Z$ (see
	Lemma~\ref{lm:algofailurecondition}). Note that the converse is not necessarily true, for example if
	there exists another close code word $\bC'\neq \bC$ with $d(\bC',\bR) \le d$ and if the SVP solver
	outputs $v_s=v_{\bC'}$.
	
	Nevertheless, we can upper bound the failure probability of the algorithm as
	$\P_{\fail}\leq\P(S_{\bR}\not\subseteq  v_{\bC}\Z)$. We introduce some notations: for
	$C\in\mathbb{R}_{>0}$ we let $\Z_{m,C} \coloneq \left\{ a \in \Z/m\Z : |a \crem m| \leq C \right\}$,
	where $a \crem m$ is the central remainder of $a$ modulo $m$, that is the unique
	representative of $a$ modulo $m$ within the interval  $\left[-\lc m/2\rc+1,\lf m/2\rf \right]$. Note
	that this set has cardinality $\# \Z_{m,C} \leq 2 \lfloor C \rfloor + 1$.
	Let $\SE$ be the set $\SE\coloneq \left\{ \varphi\in\Z/\Lambda\Z : \forall i, \ g\varphi E'_i \in
	\Z_{\Lambda,B \Lambda} \right\}$ for $B\coloneq2^{d}\beta\frac{2FG}{N}$.

	We need a new constraint to prove the following lemma.
	\begin{constraint}
		\label{c_3}
		Algorithm~\ref{algoSRN} parameters satisfy $B < 1$.
	\end{constraint}
	
	\begin{lemma}
		\label{BaseLemmaRN}
		If Constraint~\ref{c_3} is satisfied,  
		$\SE=\{0\} \Rightarrow S_{\bR}\subseteq  v_{\bC}\Z$.
	\end{lemma}
	\begin{proof}
		Let $(\varphi,\psi_1,\ldots,\psi_\ell)\in S_{\bR}=S_{\bR,2^d \beta}$. We know that for all $1 \le i \le \ell$,
		$g\varphi E_i  = g\varphi\left(R_i-\frac{f_i}{g}\right)  = g\psi_i-f_i\varphi \bmod N.$ Since
		$Y|E_i$ and $Y|N$, thanks to the above, we have that $Y|(g\psi_i-f_i\varphi)$, and we define the
		integer $\psi'_i=\frac{g\psi_i-f_i\varphi}{Y}$. Dividing the above modular equation by $Y$ we
		obtain $g\varphi E'_i = \psi'_i \bmod \Lambda$. Therefore,
		\begin{equation*}
			|g\varphi E'_i \crem \Lambda|\leq|\psi'_i|\leq\frac{|g\psi_i|+|f_i\varphi|}{Y}< 2^{d}\beta\frac{2FG}{N} \Lambda = B \Lambda
		\end{equation*} 
		which means that $\varphi\in \SE$, thus thanks to the hypothesis $\SE=\{0\}$, we get
		$\Lambda|\varphi$, thus $g\varphi E'_i = \psi'_i = 0 \bmod \Lambda.$
		Thanks to Constraint~\ref{c_3} and the above inequality we can conclude that
		$|\psi'_i|<\Lambda$, therefore $\psi'_i=0$ in $\Z$. Which means that 
		\begin{equation}
			\label{eq:samefractions}
			\forall i=1,\ldots,\ell, \ g\psi_i = f_i\varphi.
		\end{equation}
		
		Since $\gcd(f_1,\ldots,f_\ell, g)=1$, Equations~\eqref{eq:samefractions} imply that $g|\varphi$.
		We have already seen that $\Lambda|\varphi$, so $g\Lambda|\varphi$ because $g$ and $\Lambda$ are
		coprime. Plugging $\varphi = a g\Lambda$ for some $a \in \Z$ into
		Equations~\eqref{eq:samefractions}, we deduce $g\psi_i = f_i\varphi = f_i a g\Lambda$, so
		$\psi_i = a \Lambda f_i$ for all $i$. We have shown $(\varphi,\psi_1,\ldots,\psi_\ell)\in
		(\Lambda g,\Lambda f_1,\ldots,\Lambda f_\ell)\Z$.
	\end{proof}
	Thanks to the above lemma we can upper bound the failure probability of
	Algorithm~\ref{algoSRN} with $\P_{\fail}\leq \P(\SE\ne\{0\}).$
	
	\begin{remark}\label{rmk:unicity}
		We note that, when the distance parameter $d$ of the decoding algorithm is
		below the unique decoding capacity of SRN codes, \ie
		$d<\log(\sqrt{N/(2FG)})$, we must have that $B\Lambda< \beta$ since
		$\Lambda\le 2^d$. As pointed out in Remark~\ref{rem:failevenifbetais1}, it
		is not because of the approximation factor that Algorithm~\ref{algoSRN}
		might fail, thus, at the cost of using an exact SVP solver, \ie\ a
		subroutine $\mathcal{ASVP_\infty}$ returning the shortest vector of
		$\bar{\L}$, we can assume $\beta = 1$. 
		Note that polynomial  time exact SVP solver exist for constant dimension $\ell$.
		Under such circumstance we therefore
		have $\Z_{\Lambda, B\Lambda} = \Z_{\Lambda, 0} = \{0\}$, thus estimating the
		failure probability of Algorithm~\ref{algoSRN} by studying $\P(\SE\ne\{0\})$
		yields the expected unique decoding result when $d<\log(\sqrt{N/(2FG)})$.
	\end{remark}

	In order to estimate $\P(\SE\ne\{0\})$, we need the following preliminary result:

	\begin{lemma}
		\label{lm:probagpe}
		If $\varphi\in\Z$ is such that $\gcd(\varphi,\Lambda)=\eta = \prod_{j\in\xi}p_j^{\eta_j}$, then
		for the probability distribution of error model $\ERR$, we have
		$$
		\P\left(
		\forall i, \ g \varphi E_i' \in \Z_{\Lambda,B\Lambda}
		\right)
		\leq
		\frac{\left(\# \Z_{\Lambda/\eta,B\Lambda/\eta}\right)^\ell}{\left(\frac{\Lambda}{\eta}\right)^\ell\prod_{p\in\cP\left(\frac{\Lambda}{\eta}\right)} 
			\left(1 - 1 / p^\ell \right)            
		}.
		$$

		If we also suppose $B < \eta/\Lambda < 1
		$, then 
		$
		\P\left(
		\forall i, \ g \varphi E_i' \in \Z_{\Lambda,B\Lambda}
		\right) = 0.
		$
	\end{lemma}
	\begin{proof}
		Since $\gcd(g,N)=1$, the distributions of the vectors $(\varphi E_1',\ldots,\varphi E_\ell')$
		and $(g\varphi E_1',\ldots,g\varphi E_\ell')$ over the sample space $$\Omega_\Lambda \coloneq \{
		(F_i)_{1\le i \le \ell} \in (\Z/\Lambda\Z)^\ell : \gcd(F_1,\dots,F_\ell,\Lambda) = 1\},$$ are
		identical. Thus, we have 
		$
		\P( \forall i, \ g \varphi E_i' \in \Z_{\Lambda,B\Lambda} )
		=
		\P( \forall i, \ \varphi E_i' \in \Z_{\Lambda,B\Lambda} )
		$.
		
		Let us now show that $\varphi E_i' \in \Z_{\Lambda,B\Lambda} \Leftrightarrow (\varphi/ \eta) E_i' \in
		\Z_{\Lambda/\eta,B\Lambda/ \eta}$: The first condition can be rephrased as 
		$
		\varphi E_i'=a_i\Lambda + c_i
		$
		with $a_i,c_i\in\Z$ and $|c_i|\le B\Lambda$. But then we must have that $ \eta|c_i$. Thus, we can divide
		the above by $\eta$ and obtain
		$
		(\varphi/ \eta) E_i'=a_i\Lambda/ \eta + c_i/ \eta
		$
		with $|c_i/ \eta|\leq B\Lambda/ \eta$, which is equivalent to $(\varphi/ \eta) E_i' \in \Z_{\Lambda/
			\eta,B\Lambda/ \eta}$.
		
		When $B\Lambda < \eta$, the previous condition implies that $(\varphi/ \eta) E_i' = 0 \bmod \Lambda/
		\eta$ for all $i$. Since $\varphi/ \eta$ is coprime with $\Lambda/ \eta$, we have $E_i' = 0
		\bmod \Lambda/ \eta$ for all $i$. If $ \eta < \Lambda$, this is in contradiction with
		$\gcd(E_1',\dots,E_\ell',\Lambda)=1$ for all random matrix $\bE$. Therefore, the associated
		probability $\P( \forall i, \ g \varphi E_i' \in \Z_{\Lambda,B\Lambda})$ is zero. 
		
		We have seen that 
		$$\P( \forall i, \ g \varphi E_i' \in \Z_{\Lambda,B\Lambda} ) =\P(\{ \bE = (\be_j)_{1 \le j \le n}
		: \forall i, \ (\varphi/ \eta) E_i' \in \Z_{\Lambda/ \eta,B\Lambda/ \eta} \}),
		$$ 
		and since $\gcd\left(\Lambda/\eta,\varphi/\eta\right) = 1$, the above reduces to 
		$$\P(\{ \bE = (\be_j)_{1 \le j \le n} : \forall i, \  E_i' \in \Z_{\Lambda/ \eta,B\Lambda/ \eta} \}).
		$$
		Now, the condition $E_i' \in \Z_{\Lambda/ \eta,B\Lambda/\eta}$ only depends 
		on the columns $(\be'_j)$ of the reduced random matrix for 
		$j \in \xi_{\Lambda/ \eta} \coloneq \{ j : \eta_j < \nu_{p_j}(\Lambda) \}$.
		These columns are uniformly distributed in the sample space 
		$
		\bar{\Omega}_{\Lambda,\eta}.
		$
		
		Therefore, letting
		$
		\Upsilon \coloneq \Bigl\{\bE = (\be_j)_{1 \le j \le n} : \forall i, \  E_i' \in \Z_{\Lambda/ \eta,B\Lambda/ \eta}\Bigr\}
		$, we note that $\# \Upsilon = (\# \Z_{{\Lambda}/{ \eta},{B\Lambda}/{ \eta}})^\ell$, and we can deduce that our probability equals 
		$$
		\P(\Upsilon) = \frac{ \# (\bar{\Omega}_{\Lambda,\eta} \cap \Upsilon) } { \# \bar{\Omega}_{\Lambda,\eta} }
		\leq 
		\frac{ \# \Upsilon } { \# \bar{\Omega}_{\Lambda,\eta} }.
		$$
		Finally, Lemma~\ref{lm:EulerFormula} tells us that $\# \bar{\Omega}_{\Lambda,\eta} = \left(\frac{\Lambda}{\eta}\right)^\ell\prod_{p\in\cP\left(\frac{\Lambda}{\eta}\right)} 
		\left(1 - 1 / p^\ell \right)$.
	\end{proof}
	Before proving our results we still need the following technical lemma.
	
	\begin{lemma}
		\label{lm:sumOverDivisorsINT}
		Given $\Lambda \in\Z$ and $f(x,y)$ an arbitrary real-valued function of two variables. Then 
		\begin{equation*}
			\sum_{\eta|\Lambda}
			\prod_{p\in\cP(\eta)}
			f(p,\nu_{p}(\eta))
			=
			\prod_{p\in\cP(\Lambda)}
			\left[
			1 + \sum_{k= 1}^{\nu_p(\Lambda)}f(p,k)
			\right]
		\end{equation*}
	\end{lemma}
	\begin{proof}
		By expanding the product on the right-hand side we obtain
		\begin{equation*}
			\prod_{p\in\cP(\Lambda)}
			\left[
			1 + \sum_{k= 1}^{\nu_p(\Lambda)}f(p,k)
			\right]
			=
			\sum_{S\subseteq\cP(\Lambda)}
			\sum_{\substack{\left(\eta_p\right)_{p\in S}\\ 1\le\eta_p\le\nu_{p}(\Lambda)}}
			\prod_{p\in S}
			f(p,\eta_p).
		\end{equation*}
		The double sum above corresponds exactly to a single sum over the divisors $\eta$ of $\Lambda$ with $S = \cP(\eta)$ and $\eta_p = \nu_{p}(\eta)$.
	\end{proof}
	
	Rewriting $\{\bE : \SE \neq \{0\} \}$ as 
	$\cup_{\varphi=1}^{\Lambda-1}\{\bE : \varphi \in \SE \}$, we get
	\begin{equation}
		\label{fstUpperBoundProof}
		\P(\SE\ne\{0\})
		\leq
		\sum_{\varphi=1}^{\Lambda-1}
		\P
		\left(
		\forall i, \ g\varphi E_i' \in \Z_{\Lambda, B\Lambda}
		\right)   
		=
		\sum_{\varphi=1}^{\Lambda-1}
		\P
		\left(
		\forall i, \ \varphi E_i' \in \Z_{\Lambda, B\Lambda}
		\right)   
	\end{equation}
	where the last equality comes from the proof of Lemma~\ref{lm:probagpe}.
	We analyze the latter quantity in the following lemma.
	
	\begin{lemma}
		\label{lm:boundfailureprobability}
		Given a random vector $(E_1',\ldots,E_\ell')$ uniformly distributed in $\Omega_{\Lambda}$, we have that
		\begin{equation*}
			\sum_{\varphi=1}^{\Lambda-1}
			\P
			\left(
			\forall i, \ \varphi E_i' \in \Z_{\Lambda, B\Lambda}
			\right) 
			\leq
			\left(3B\right)^\ell \Lambda\prod_{p\in\cP(\Lambda)}\left(\frac{1 - 1/p^{\ell + \nu_p(\Lambda)}}{1 -1 / p^{\ell}}\right).
		\end{equation*}
	\end{lemma}
	\begin{proof}
		We can use Lemma~\ref{lm:probagpe} and upper bound the terms in the sum with 
		$$
		\P \left( \forall i, \ \varphi E_i' \in \Z_{\Lambda, B\Lambda} \right)
		\le
		\frac{\left(\#
			\Z_{\Lambda/\eta,B\Lambda/\eta}\right)^\ell}{\left(\frac{\Lambda}{\eta}\right)^\ell\prod_{p\in\cP\left(\frac{\Lambda}{\eta}\right)}
			\left(1 - \frac{1}{p^\ell}\right)            
		}
		$$
		where $\eta=\gcd(\varphi,\Lambda)$. Thanks to the second point in Lemma~\ref{lm:probagpe}, we can restrict the sum only to the elements $\varphi$ such that $\eta\leq B\Lambda$, which in turn allows us to deduce that 
		$
		\#\Z_{\Lambda/\eta,B\Lambda/\eta} \leq 2 \lfloor B\Lambda/\eta \rfloor + 1 \leq 3 B\Lambda/\eta.
		$
		Since this expression depends only on $\eta$, we regroup the $\varphi$ in the sum by their gcd with $\Lambda$.  
		Note that the number of elements $\varphi\in \Z_\Lambda$ such that $\gcd(\varphi,\Lambda)=\eta$, is equal to $\phi\left(\frac{\Lambda}{\eta}\right)$ with $\phi$ being the Euler's totient function. 
		Therefore,
		$$
		\sum_{\substack{\varphi=1\\\eta=\gcd(\varphi,\Lambda)\leq B\Lambda}}^{\Lambda -1} \frac{\left(\#
			\Z_{\Lambda/\eta,B\Lambda/\eta}\right)^\ell}{\left(\frac{\Lambda}{\eta}\right)^\ell\prod_{p\in\cP\left(\frac{\Lambda}{\eta}\right)}
			\left(1 - 1/p^\ell\right)            
		}
		\leq 
		\sum_{\substack{\eta|\Lambda\\\eta\leq B\Lambda}}
		\frac{\phi\left(\frac{\Lambda}{\eta}\right)\left(\frac{3B\Lambda}{\eta}\right)^\ell}
		{\left(\frac{\Lambda}{\eta}\right)^\ell\prod_{p\in\cP\left(\frac{\Lambda}{\eta}\right)} \left(1
			- 1/p^\ell\right)}.
		$$
		Extending the sum over all the divisors $\eta$, we can upper bound the quotient $\P(\SE\ne\{0\})/\left(3B\right)^\ell$ with
		\begin{align*}
			\sum_{\eta|\Lambda}
			\frac{\phi\left(
				\frac{\Lambda}{\eta}
				\right)
			}
			{\prod_{p\in\cP\left(\frac{\Lambda}{\eta}\right)} 
				\left(1 - \frac{1}{p^\ell}\right)}
			=
			\sum_{\eta|\Lambda}
			\prod_{p\in\cP(\eta)}
			\frac{1 - \frac{1}{p}}{1 - \frac{1}{p^{\ell}}}
			p^{\nu_{p}\left(\eta\right)}
			=
			\prod\limits_{p \in\cP(\Lambda)}
			\left(
			1+
			\frac{1 - \frac{1}{p}}{1 - \frac{1}{p^{\ell}}}
			\sum_{k = 1}^{\nu_p(\Lambda)}
			p^{k}
			\right)
		\end{align*}
		where in the last equality we used Lemma~\ref{lm:sumOverDivisorsINT} with 
		\begin{equation*}
			f(x,y)
			=
			\frac{1 - \frac{1}{x}}{1 - \frac{1}{x^\ell}}
			x^{y}.
		\end{equation*}
		To conclude we notice that
		\begin{equation*}
			\prod\limits_{p \in\cP(\Lambda)}
			\left(
			1+
			\frac{1 - \frac{1}{p}}{1 - \frac{1}{p^{\ell}}}
			\sum_{k = 1}^{\nu_p(\Lambda)}
			p^{k}
			\right)
			=
			\prod_{p\in\cP(\Lambda)}
			\frac{p^{\nu_p(\Lambda)} - 1 / p^{\ell}}{1 - 1 / p^\ell}
			=
			\Lambda
			\prod_{p\in\cP(\Lambda)}
			\frac{1 - 1 / p^{\ell + \nu_p(\Lambda)}}{1 - 1 / p^\ell}. \qedhere
		\end{equation*}
	\end{proof}		
	
	\begin{proof}[Proof of Theorem~\ref{thm:main1}]
		We start by proving that any choice of the input parameter $d\leq \dm$ satisfies  Constraint~\ref{c_3}, thus we can apply all the previous lemmas and upper bound the failure probability of Algorithm~\ref{algoSRN} with the quantity given by Lemma~\ref{lm:boundfailureprobability}. Remark that
		\begin{equation*}
			2\beta\frac{2^{d}FG}{N}\leq 2\beta\frac{2^{\dm}FG}{N} 
			= 
			\frac{2\beta FG}{N}\left(\frac{N}{6FG\beta}\right)^{\frac{\ell}{\ell +1}}
			=
			\left( \frac{2FG}{N} \frac{\beta}{3^\ell} \right)^{\frac{1}{\ell +1}}. 
		\end{equation*}
		We already noticed when defining the $\SRN_{\ell}(N;F,G)$ code that $2FG<N$.
		Thanks to Constraint~\ref{cst:beta}, we know that $\beta < 3^\ell$, thus
		the above quantity is smaller than $1$ and Constraint~\ref{c_3} is satisfied.
		
		As noticed in Equation~\eqref{fstUpperBoundProof}, 
		$
		\P(\SE\ne\{0\}) 
		\leq \sum_{\varphi=1}^{\Lambda-1}
		\P
		\left(
		\forall i, \ \varphi E_i' \in \Z_{\Lambda, B\Lambda}
		\right)$,
		which we can upper bound using Lemma~\ref{lm:boundfailureprobability}.
		Thanks to the hypothesis of Theorem~\ref{thm:main1} we know that $\Lambda \leq 2^{d}$, 
		and using 
		$\left(3B\right)^\ell 2^d  = 2^{-(\ell + 1)(\dm - d)}$, 
		we have proved Theorem~\ref{thm:main1}.
	\end{proof}
	
	\subsubsection{Decoding failure probability with respect to the second error model}
	\label{RN_ERR2}
	
	In the second error model, we need to make a distinction between the maximal error locator $\Lambda_m$ (over which there are uniform random errors) and the actual error locator $\Lambda$ which is in general a divisor of $\Lambda_m$. We will denote $\P_{\ERRd}$ (resp. $\P_{\ERR}$) the probability function under the error model 2 (resp. the error model 1). 
	Let $\cF$ be the event of decoding failure with algorithm parameter $d\geq\log(\Lambda_m)$ \ie\  the set of random matrices $\bE$ such that Algorithm~\ref{algoSRN} returns "decoding failure".
	Using the law of total probability, we have
	\begin{equation}
		\label{LawTotalProba}
		\P_{\ERRd}( \cF )
		=
		\sum_{\Lambda|\Lambda_m}
		\P_{\ERRd}(\cF \ |\ \Lambda_{\bE} = \Lambda)
		\ 
		\P_{\ERRd}( \Lambda_{\bE} = \Lambda)
	\end{equation}
	where $\Lambda_{\bE} =\Lambda_{\bC,\bR}$ (see Definition~\ref{def:IntDistance}).
	The conditional probabilities 
	$\P_{\ERRd}(\cF \ |\ \Lambda_{\bE} = \Lambda)$
	in the sum  are equal to 
	$\P_{\ERR}(\cF)$, which are upper bounded within the proof of Lemma~\ref{lm:boundfailureprobability} by
	\begin{equation}
		\label{IntermediateBuondERR1}
		\P_{\ERR}(\cF)
		\leq
		\left(3B\right)^\ell \Lambda\prod_{p\in\cP(\Lambda)}\left(\frac{1 - 1 / p^{\ell + \nu_p(\Lambda)}}{1 -1 / p^{\ell}}\right).
	\end{equation}
	Moreover, using again Lemma~\ref{lm:EulerFormula}, we have 
	\begin{equation}
		\label{ConditionProba}
		\P_{\ERRd}( \Lambda_{\bE} = \Lambda)
		=
		\frac{\#\Omega_\Lambda}{\Lambda_m^\ell}
		=
		\left(\frac{\Lambda}{\Lambda_m}\right)^{\ell}
		\prod_{p\in\cP(\Lambda)}\left(1 - \frac{1}{p^{\ell}}\right).
	\end{equation}
	Using these facts we can prove Theorem~\ref{thm:main2}.

	\begin{proof}[Proof of Theorem~\ref{thm:main2}]
		Plug Equations~\eqref{ConditionProba} and~\eqref{IntermediateBuondERR1} in Equation~\eqref{LawTotalProba} to obtain that 
		$\P_{\ERRd}( \cF ) / (\frac{3B}{\Lambda_m})^\ell$ is less than or equal to
		\begin{align*}
			\sum_{\Lambda|\Lambda_m}
			\Lambda^{\ell + 1}
			\prod_{p \in\cP(\Lambda)}\left(1 -\frac{1}{p^{\ell + \nu_p(\Lambda)}}\right) 
			&=
			\sum_{\Lambda|\Lambda_m}
			\prod_{p\in\cP(\Lambda)}p^{\nu_p(\Lambda)(\ell + 1)}\left(1 -\frac{1}{p^{\ell + \nu_p(\Lambda)}}\right)\\
			&=
			\prod_{p\in\cP(\Lambda_m)}\left[1 + \sum_{k = 1}^{\nu_p(\Lambda_m)}p^{k(\ell + 1)}\left(1 -\frac{1}{p^{\ell + k}}\right)\right]\\
			&\le
			\prod_{p\in\cP(\Lambda_m)}\left[1 + \left(1 -\frac{1}{p^{\ell + \nu_p(\Lambda_m)}}\right)\sum_{k = 1}^{\nu_p(\Lambda_m)}p^{k(\ell + 1)}\right],
		\end{align*}
		where we used again Lemma~\ref{lm:sumOverDivisorsINT} with 
		$
		f(x,y)=
		x^{y(\ell + 1)}\left(1 -\frac{1}{x^{\ell + y}}\right),
		$
		and in the last inequality we used that $1 -1/p^{\ell + k}\leq 1 -1/p^{\ell + \nu_p(\Lambda_m)}$ for every $k=1,\ldots,\nu_p(\Lambda_m)$. 
		By computing the geometric sum inside the last product, the above is equal to
		\begin{align*}
			&\prod_{p\in\cP(\Lambda_m)}
			\left[
			1 + \left(1 - \frac{1}{p^{\ell + \nu_p(\Lambda_m)}}\right)
			\left(\frac{p^{(\ell + 1)(\nu_p(\Lambda_m) + 1)} - 1}{p^{\ell + 1} - 1} - 1\right)
			\right]\\
			&=
			\prod_{p\in\cP(\Lambda_m)}
			\left[
			1 +
			\frac{1 - 1 / p^{\ell + \nu_p(\Lambda_m)}}{1 - 1 /p^{\ell + 1}}
			\left(p^{\nu_p(\Lambda_m)(\ell + 1)} - 1\right)
			\right].
		\end{align*}
		Since $\nu_p(\Lambda_m) \geq 1$ we have that 
		$
		1 \leq (1 - 1/p^{\ell + \nu_p(\Lambda_m)})/(1 - 1/p^{\ell + 1})
		$
		and the above product is upper bounded as:
		\begin{equation*}
			\prod_{p\in\cP(\Lambda_m)}
			\left[
			1 +
			\frac{1 - 1/p^{\ell + \nu_p(\Lambda_m)}}{1 - 1/p^{\ell + 1}}
			\left(p^{\nu_p(\Lambda_m)(\ell + 1)} - 1\right)
			\right]
			\le
			\Lambda_m^{\ell + 1}
			\prod_{p\in\cP(\Lambda_m)}
			\frac{1 - 1/p^{\ell + \nu_p(\Lambda_m)}}{1 - 1/p^{\ell + 1}}
		\end{equation*}

		Now, thanks to the hypothesis of the theorem 
		we know that $\Lambda_m\leq 2^{d}$, thus we can write 
		\begin{align*}
			\P_{\xi_r}^{\ERRd}( \cF ) &\leq \left(3B\right)^\ell	\Lambda_m
			\prod_{p\in\cP(\Lambda_m)}
			\frac{1 - 1/p^{\ell + \nu_p(\Lambda_m)}}{1 - 1/p^{\ell + 1}}\\
			&\leq 
			\left(3B\right)^\ell 2^{d}
			\prod_{p\in\cP(\Lambda_m)}
			\frac{1 - 1/p^{\ell + \nu_p(\Lambda_m)}}{1 - 1/p^{\ell + 1}}.
		\end{align*}
		Using $2^{-(\ell+1)(\dm - d)}=\left(3B\right)^\ell 2^{d}$,
		we have proved Theorem~\ref{thm:main2}.
	\end{proof}

	\section{Analysis of the decoder for a hybrid error model}
	\label{sec:Hybrid Decoding}

	In this section we consider a hybrid approach to the failure probability analysis for the multiplicity rational codes studied above.
	The approach is hybrid in the sense that it lies in between unique decoding and interleaving.

	More specifically, in the algorithm, the parameter $d$ is chosen, and it is strictly related to the failure probability. In the analysis, $d$ splits into two components: $d_{i}$ and $d_{u}$. Essentially, $d_{u}$ is bounded for fitting the unique decoding, whereas $d_{i}$ can be larger as it is related to the interleaving decoding and its bound $\dmi$ (Equation~\eqref{def:DmaxHyb}) is directly proportional to the parameter $\ell$.
	Notably, if $d_{i}=0$, the algorithm never fails. Therefore, the probability of decoding failure is strictly related to $d_{i}$ and is analyzed under probabilistic assumptions, particularly considering a random error distribution. 
	
	The motivation for splitting $d$ is that not all errors can be assumed to be purely random. 
	For instance, in the context of distributed computation, some errors might be introduced by malicious entities that deliberately choose specific error patterns to force the algorithm to fail. In such cases, the errors captured by $d_{u}$ remain independent of the error distribution and can still be corrected.
	
	Since we are above the unique decoding radius, not all errors are decodable. Interleaving techniques can provide positive decoding results by considering error sets where most errors are decodable using probabilistic arguments. These techniques focus on fixed error positions and consider all possible errors at each position.
	In contrast, in a hybrid setting one can handle more general sets of errors, analyzing the set of all possible errors across certain subsets of the error positions.
	This approach may be of broader interest in coding theory.
	
	A first instance of this hybrid model has been introduced in~\cite{guerrini2023simultaneous} in the context of rational function reconstruction with multiplicities and poles. 
	Then, the hybrid model has been jointly generalized in the present work and in~\cite{brakensiek2025unique} where the authors studied the decoding properties of interleaved and folded Reed Solomon as well as multiplicity codes under this hybrid error model. 
	In our context the motivation of introducing such model will become clear in the forthcoming case of codes allowing bad primes (See Section~\ref{Sec:bad primes}). In that case, the only result we are able to get is when we only interleave a subset of all errors (namely evaluation errors).
	We remark that, as in~\cite{guerrini2023simultaneous} for the rational function case and a different analysis, with the hybrid technique we are only able to interleave a specific type of errors. This suggests us that there could be a deeper obstacle preventing us to interleave the other type of errors (namely valuation errors).
	
	On a technical level this hybrid analysis consists in studying the failure probability with respect to a specific portion of the error's distribution; allowing the errors to vary only over a subset $\xii\subseteq\xi$ of the  error support, while the errors in the complementary set $\xiu\coloneq\xi\setminus\xii$ are held fixed. 
	Note that, in this section the above partition might seem arbitrary but, as we will see in the next Section~\ref{Sec:bad primes} on bad primes, it is clearly described by some property of the error itself (see Definition~\ref{def:errorsupportbad primes}).
	Here we generalize the analysis of the previous section relative to the decoding of $\SRN$ codes (Definition~\ref{def:SRN}) by means of Algorithm~\ref{algoSRN}.
	In this setting we decompose the distance parameter $d$ of the algorithm as
	\begin{equation}
		\label{distParamAlgoHYB}
		d = \hdi + \hdu,
	\end{equation}
	for some  $\hdi,\hdu \geq 0$ bounds on the sizes of random and fixed errors respectively.

	\paragraph*{Error models}
	With the given distance parameter $d$ as in Equation~\eqref{distParamAlgoHYB}, 
	we perform the hybrid analysis with respect to a distribution
	specified by a  factorization of $\Lambda = \Lu\Li$ with $\gcd(\Lu,\Li) = 1$, $\Lambda$ divides $N$. To specify the error model, we fix a sequence of nonzero error vectors 
	$\beps_j\in\left(\Z/p_j^{\lambda_j}\Z\right)^{\ell}$ for every $j$ 
	such that $p_j\in\cP(\Lu)$, with $\nu_{p_j}(\beps_j) = \lambda_j - \nu_{p_j}(\Lambda)$. 
	Then the random distribution for the hybrid error model is determined by the set of error matrices 
	$\bE\in\prod_{j=1}^{n}\left(\Z/p_j^{\lambda_j}\Z\right)^{\ell}$ such that the columns $\be_j$ of $\bE$ satisfy
	\begin{enumerate}
		\item $\be_j = \boldsymbol{0} \text{ for all } j \text{ such that } p_j\not\in\cP(\Lambda)$,
		\item $\be_j = \beps_j \ \text{ for all } j\text{ such that } p_j\in\cP(\Lu)$,
		\item $\nu_{p_j}(\be_j) = \lambda_j - \nu_{p_j}(\Lambda) \text{ for all } j \text{ such that } p_j\in\cP(\Li)$.
	\end{enumerate}
	We let $\HYB$ be the set of error matrices specified as above.

	\begin{lemma}
		\label{lm:randomuniform}
		If $\bE$ is uniformly distributed in $\HYB$, then 
		the random vector $(E_1' \bmod \Li,\ldots,E_\ell' \bmod \Li)$ is uniformly
		distributed in the sample space $\Omega_{\Li}$.
	\end{lemma}
	\begin{proof}
		For the duration of this proof, we will only consider indices $j$ such that
		$p_j\in\cP(\Li)$. Recall that  $\be_{j}$ is a random vector of 
		$\left(\Z/p_j^{\lambda_j}\Z\right)^{\ell}$ of valuation $\lambda_j - \nu_{p_j}(\Lambda)$
		for all those particular $j$. Since $Y = p_j^{\lambda_j - \nu_{p_j}(\Lambda)} \bmod p_j^{\lambda_j}$,
		we get that $E'_i = E_i/Y = e_{i,j}/Y \bmod p_j^{\nu_{p_j}(\Lambda)}$. 
		By definition of $\HYB$, the vector $\be_{j}/Y \in
		(\Z/p_j^{\nu_{p_j}(\Lambda)}\Z)^\ell$ is random of valuation $0$. As a
		consequence, we obtain that $(E_1' \bmod \Li,\ldots,E_\ell' \bmod \Li)$ is random
		among the vectors of $(\Z/\Li\Z)^\ell$ such that
		$\gcd(E_1',\ldots,E_{\ell}',\Li) = 1$.
	\end{proof}
	
	\medskip

	Whereas for the hybrid version of the error model $\ERRd$, we fix a maximal error locator $\Lambda_m$ factorized as $\Lambda_m = \Lmi\Lu$ with $\gcd\left(\Lmi,\Lu\right) = 1$. We fix a sequence of nonzero error vectors $\beps_j\in\left(\Z/p_j^{\lambda_j}\Z\right)^{\ell}$ for every $j$ such that $p_j\in\cP(\Lu)$, with $\nu_{p_j}(\beps_j) = \lambda_j - \nu_{p_j}(\Lambda_m)$. Then we consider the set of error matrices $\bE\in\prod_{j=1}^{n}\left(\Z/p_j^{\lambda_j}\Z\right)^{\ell}$ such that
	\begin{enumerate}
		\item $\be_j = \boldsymbol{0} \text{ for all } j \text{ such that } p_j\not\in\cP(\Lambda_m)$,
		\item $\be_j = \beps_j \ \text{ for all } j\text{ such that } p_j\in\cP(\Lu)$,
		\item $\nu_{p_j}(\be_j) \ge \lambda_j - \nu_{p_j}(\Lambda_m) \text{ for all } j \text{ such that } p_j\in\cP(\Lmi)$.
	\end{enumerate}
	We let $\HYBd$ be the set of error matrices specified as above.
	
	We notice that for a given error matrix in the distribution $\HYBd$ the associated error locator has the form $\Lambda = \Li\Lu$ for some divisor $\Li|\Lmi$.

	\paragraph*{Our results} We can now state our results concerning the analysis of the correctness of the decoder \wrt to a hybrid error model. Define
	\begin{equation}
		\label{def:DmaxHyb}
		\dmi \coloneq 
		\frac{\ell}{\ell+1} \left[\log(N/2FG) - \log(3\beta) - 2 \hdu\right].
	\end{equation}
	
	Note that we must have $2 \hdu \le \log(N/(6FG \beta))$ in order to ensure $\dmi \ge 0$.

	\begin{theorem}
		\label{thm:main1hyb}
		Decoding Algorithm~\ref{algoSRN} on input
		\begin{enumerate}
			\item distance parameter $d = \hdu + \hdi$ for $\hdu\leq \log\left(\sqrt{N/(6FG \beta)}\right)$ and $\hdi\leq \dmi$,
			\item a random received word $\bR$ uniformly distributed in  $\left[\bf/g\right]_N + \HYB$
			for some code word $\left[\bf/g\right]_N \in \SRN_\ell(N;F,G)$ and error locator $\Lambda = \Li\Lu$ such that $\log (\Lu) \leq \hdu$ and $\log(\Li)\leq \hdi$,
		\end{enumerate}
		outputs the center  code word $\left[\bf/g\right]_N$ of the distribution
		with a probability of failure 
		\begin{equation*}
			\P_{\fail}
			\leq 
			2^{-(\ell+1)(\dmi - \hdi)}\prod_{p\in\cP(\Li)}\left(\frac{1 - 1/p^{\ell + \nu_p(\Li)}}{1 -1/p^{\ell}}\right).
		\end{equation*} 
	\end{theorem}
	\begin{theorem}
		\label{thm:main2hyb}
		Decoding Algorithm~\ref{algoSRN} on input
		\begin{enumerate}
			\item distance parameter $d = \hdu + \hdi$ for $\hdu\leq \log\left(\sqrt{N/(6FG \beta)}\right)$ and $\hdi\leq \dmi$,
			\item a random received word $\bR$ uniformly distributed in  $\left[\bf/g\right]_N + \HYBd$
			for some code word $\left[\bf/g\right]_N \in \SRN_\ell(N;F,G)$ and error locator $\Lambda_m = \Lmi\Lu$ such that $\log (\Lu) \leq \hdu$ and $\log(\Lmi)\leq \hdi$,
		\end{enumerate}
		outputs the center  code word $\left[\bf/g\right]_N$ of the distribution
		with a probability of failure 
		\begin{equation*}
			\P_{\fail}
			\leq 
			2^{-(\ell+1)(\dmi - \hdi)}
			\prod_{p\in\cP(\Lmi)}
			\left(\frac{1 - 1/p^{\ell + \nu_p(\Lmi)}}{1 - 1/p^{\ell + 1}}\right).
		\end{equation*} 
	\end{theorem}
	
	\begin{exam}\label{example}
		Let's give a scenario that would highlight how Theorem~\ref{thm:main2hyb}
		can be used in practice. Assume that a code is fixed such that 
		$\log(N/(6FG\beta)) = 200$, so that $\dm = 160$ when one interleaves for $\ell=4$.
		Assume one wanted to make sure that the failure probability is less than a target probability of $2^{-30}$, and also that $50$ weighted errors can always be corrected ($\hdu = 50$), for instance for protecting against a malicious entity.
		Then $\dmi = 80$ and one would have to choose the parameter $d = 134$ (thus $\hdi = 74$) for the decoder (where we approximate the failure probability by $2^{-(\ell+1)(\dmi - \hdi)}$).
		Then Theorem~\ref{thm:main2hyb} would ensure that for any error with locator 
		$\Lu$ such that $\log \Lu \le 50$ and for any random error distributed uniformly
		on an error locator $\Lmi$ such that $\log \Lmi \le 74$ (with $\Lmi$ and $\Lu$ coprime), the failure probability is less than $2^{-30}$.
	\end{exam}

	We introduce a modified version of the set $\SE$ defined as
	\begin{equation*}
		\SEh\coloneq
		\left\{
		\varphi\in\Z/\Li\Z : \forall i, \ g\varphi E'_i \in \Z_{\Li,B\Lambda}
		\right\}
	\end{equation*}
	with $B\coloneq2^{d}\beta\frac{2FG}{N} = 2^{\hdi+\hdu}\beta\frac{2FG}{N} $. 
	The hybrid versions of Constraint~\ref{c_3} and Lemma~\ref{BaseLemmaRN} are as follows:
	
	\begin{constraint}
		\label{cst: ParamHybridVers}
		The parameters of Algorithm~\ref{algoSRN} satisfy $2^{\hdu} B < 1$.
	\end{constraint}
	\begin{lemma}
		\label{BaseLemmaRNHybrid}
		If Constraint~\ref{cst: ParamHybridVers} is satisfied then $\SEh = \{0\}
		\Rightarrow S_{\bR}\subseteq v_{\bC}\Z$.
	\end{lemma}
	\begin{proof}
		Let $(\varphi,\psi_1,\ldots,\psi_\ell)\in S_{\bR}$. 
		The proof of Lemma~\ref{BaseLemmaRN} shows that $g\varphi E'_i$ is equal to 
		$\psi'_i \coloneq \frac{g\psi_i-f_i\varphi}{Y}$ modulo $\Lambda$, hence also modulo $\Li$.
		The same proof gives $|\psi'_i|  \leq B\Lambda$.
		This means that $\varphi\in \SEh$, thus thanks to the hypothesis $\SEh=\{0\}$, we get
		$\Lambda_i|\varphi$, thus $g\varphi E'_i = \psi'_i = 0 \bmod \Li.$
		Since $\Lu\leq 2^{\hdu}$, this implies $|\psi'_i|  \leq B\Lambda < \Li$, therefore $\psi'_i=0$ in $\Z$.
		The end of the proof is identical to the one of Lemma~\ref{BaseLemmaRN}.
	\end{proof}
	
	As in Equation~\eqref{fstUpperBoundProof}, we have 
	$
	\P(\SEh\ne\{0\})
	\leq
	\sum_{\varphi=1}^{\Li-1}
	\P
	\left(
	\forall i, \ \varphi E'_i \in \Z_{\Li,B\Lambda}
	\right)   
	$,
	which we now bound.
	
	\begin{lemma}
		\label{lm:boundfailureprobabilityhybrid}
		Given a random vector $(E_1',\ldots,E_\ell')$ uniformly distributed in $\Omega_{\Li}$, we have that
		\begin{equation*}
			\sum_{\varphi=1}^{\Li-1}
			\P
			\left(
			\forall i, \ \varphi E'_i \in \Z_{\Li,B\Lambda}
			\right)  
			\leq
			\left(
			3B 2^{\hdu}\right)^\ell \Li
			\prod_{p\in\cP(\Li)}\left(\frac{1 - 1 / p^{\ell + \nu_p(\Li)}}{1 -1 / p^{\ell}}
			\right).
		\end{equation*}
	\end{lemma}
	\begin{proof}
		As in the proof of Lemma~\ref{lm:boundfailureprobability}, we can upper bound the probability 
		$
		\sum_{\varphi=1}^{\Li-1}
		\P
		\left(
		\forall i, \ \varphi E'_i \in \Z_{\Li,B\Lambda}
		\right)  
		$
		with
		\begin{align*}
			\sum_{\substack{\eta|\Lambda_i\\\eta\leq B\Lambda}}
			\frac{
				\phi\left(\frac{\Lambda_i}{\eta}\right)\left(\frac{3B\Lambda}{\eta}\right)^\ell
			}{
				\left(\frac{\Lambda_i}{\eta}\right)^\ell\prod_{p\in\cP\left(\frac{\Lambda_i}{\eta}\right)} \left(1
				- 1/p^\ell\right)
			}
			& \leq 
			\left(3B\Lu\right)^\ell \Li
			\prod_{p\in\cP(\Li)}\left(\frac{1 - 1 / p^{\ell + \nu_p(\Li)}}{1 -1 / p^{\ell}}
			\right).
		\end{align*}
		Using that $\Lu\leq 2^{\hdu}$ we obtain our statement.
	\end{proof}

	\begin{proof}[Proof of Theorem~\ref{thm:main1hyb}]
		As in the proof of Theorem~\ref{thm:main1}, we start by noticing that our choice of parameters satisfy Constraint~\ref{cst: ParamHybridVers}. We first notice that $\hdi\leq \dmi = \ell/(\ell+1) (\log(N/2FG) - \log(3\beta) - 2 \hdu)$, thus
		\begin{align*}
			2^{2\hdu + \hdi}\frac{2\beta FG}{N}
			\le 
			2^{2\hdu + \dmi}\frac{2\beta FG}{N}
			& =
			\left( \frac{N}{6\beta FG 2^{2\hdu}}\right)^{\frac{\ell}{(\ell+1)}} \frac{2\beta FG 2^{2 \hdu}}{N} \\
			& =
			\left(\frac{2FG 2^{2 \hdu}}{N} \frac{\beta}{3^\ell} \right)^{\frac{1}{\ell +1}}.
		\end{align*}
		Since $\hdu \leq \log\left(\sqrt{N/2FG}\right)$, the fraction $2FG 2^{2 \hdu}/N$ is less than or equal to $1$.
		Thanks to Constraint~\ref{cst:beta}, we know that $\beta < 3^\ell$, thus
		the above quantity is less than $1$ and Constraint~\ref{cst: ParamHybridVers} is satisfied.
		Thanks to Lemma~\ref{BaseLemmaRNHybrid} and Lemma~\ref{lm:boundfailureprobabilityhybrid}, we can upper bound the failure probability by
		\begin{align*}
			\P_{fail}
			\leq
			\P(\SEh\ne\{0\})
			&\leq
			\sum_{\varphi=1}^{\Li-1}
			\P
			\left(
			\forall i, \ \varphi E'_i \in \Z_{\Li,B\Lambda}
			\right)\\
			&\leq
			\left(
			3B 2^{\hdu}\right)^\ell \Li
			\prod_{p\in\cP(\Li)}\left(\frac{1 - 1 / p^{\ell + \nu_p(\Li)}}{1 -1 / p^{\ell}}\right).
		\end{align*}
		Since $\Li\leq 2^{\hdi}$, we have
		$
		\left(3B 2^{\hdu}\right)^\ell \Li 
		\le 
		\left(3B\right)^\ell 2^{\ell \hdu} 2^{\hdi}
		=
		2^{-(\ell+1)(\dmi - \hdi)}
		$.
	\end{proof}

	\begin{proof}[Proof of Theorem~\ref{thm:main2hyb}]
		Let $\cF$ be the event of decoding failure, \ie{}the set of random matrices $\bE$ such that Algorithm~\ref{algoSRN} returns "decoding failure" with input parameter $d = \hdi + \hdu$ as in the statement of Theorem~\ref{thm:main2hyb}.
		We will denote $\P_{\HYBd}$ (resp. $\P_{\HYB}$) the probability function under the hybrid error model 2 (resp. model 1) specified by a given factorization of the error locator, and by a sequence of fixed error vectors $\beps_j$ for every $j$ such that $p_j\in\cP\left(\Lu\right)$.

		Using the law of total probability, we have that $\P_{\HYBd}( \cF )$ can be decomposed as the sum
		
		$$
		\P_{\HYBd}(\cF)
		=
		\sum_{\Li|\Lmi}
		\P_{\HYBd}(\cF \ |\ \Lambda_{\bE} = \Li\Lu)
		\ 
		\P_{\HYBd}( \Lambda_{\bE} = \Li\Lu),
		$$
		where 
		$
		\P_{\HYBd}(\cF \ |\ \Lambda_{\bE} = \Li\Lu)
		=
		\P_{\HYB}(\cF)
		$,
		whereas
		$$
		\P_{\HYBd}(\Lambda_{\bE} = \Li\Lu)
		= 
		\left(\frac{\Li}{\Lmi}\right)^{\ell}
		\prod_{p\in\cP(\Li)}\left(1 - \frac{1}{p^{\ell}}\right)
		$$
		as in Equation~\eqref{ConditionProba}.
		
		Plugging the above two expressions in the decomposition from the law of total probability, similarly as done in the proof of Theorem~\ref{thm:main2}, we can upper bound 
		$
		\P_{\HYBd}(\cF)/ \left( \frac{2^{\du}3 B}{\Lmi} \right)^\ell
		$ 
		by
		\begin{equation*}
			\sum_{\Li|\Lmi}
			\Li^{\ell + 1}
			\prod_{p \in\cP(\Li)}\left(1 -\frac{1}{p^{\ell + \nu_p(\Li)}}\right)
			\leq
			\Lmi^{\ell + 1}
			\prod_{p\in\cP(\Lmi)}
			\frac{1 - 1/p^{\ell + \nu_p(\Lmi)}}{1 - 1/p^{\ell + 1}}.
		\end{equation*}
		Thus,
		\begin{align*}
			\P_{\HYBd}( \cF )
			&\leq
			\left(2^{\du} 3 B \right)^\ell  \Lmi
			\prod_{p\in\cP(\Lmi)}
			\frac{1 - 1/p^{\ell + \nu_p(\Lmi)}}{1 - 1/p^{\ell + 1}}\\
			&\leq 
			\left(2^{\du} 3 B \right)^\ell   2^{\hdi}
			\prod_{p\in\cP(\Lmi)}
			\frac{1 - 1/p^{\ell + \nu_p(\Lmi)}}{1 - 1/p^{\ell + 1}}
		\end{align*}
		and we conclude by using that 
		$
		\left(2^{\du} 3 B \right)^\ell 2^{\hdi}
		=
		2^{-(\ell+1)(\dmi - \hdi)}
		$.
	\end{proof}
	
	\section{The case of bad primes}
	\label{Sec:bad primes}
	In this section we use the hybrid analysis technique presented above to extend our study of the decoding failure in a context where the hypothesis $\gcd(g,N) = 1$ of Definition~\ref{def:SRN} does not hold, thus some reductions in the encoding of $\bf/g$ may not be defined.
	Primes relative to undefined reductions are called \textit{bad primes}. 
	Our theorems~(\ref{thm:main1bad primes} and~\ref{thm:main2bad primes}) are the first results (in the bad primes' scenario) relative to the decoding beyond uniqueness for rational number codes.
	
	In the case of simultaneous rational function reconstruction (over $\F_q[x]$), instead of the primes $p_1,\ldots,p_n$, distinct evaluation points $\alpha_1,\ldots,\alpha_n\in\F_q$ are chosen, so in the polynomial context the notion of bad primes correspond to \textit{poles} of the vector of rational functions $\bf/g$, \ie roots of the denominator $g$.
	We find two approaches in the literature to deal with poles: in~\cite{kaltofen2020hermite} an extra symbol $\infty$ is used,
	while in~\cite{guerrini2023simultaneous} coordinates are given by shifted Laurent series representations of the fractions.
	
	In particular, considering the rational function case, the authors of~\cite{guerrini2023simultaneous} introduced the following multi-precision encoding composed of a valuation part and a reduction part, which we give in the rational number case, where the shifted Laurent series is replaced by a shifted $p$-adic expansion of the vector of fractions $\bf/g$: 
	\begin{defi}[Multi-precision encoding]
		\label{def:MultiprecisionEncoding}
		Given a sequence of multiplicities $\lambda_1,\ldots,\lambda_n$ associated to the primes $p_1,\ldots,p_n\in\Z$, and a reduced vector of fractions $\bf/g\in\Q^\ell$,  we define its multi-precision encoding to be the sequence of couples $\Evi \left(\bf/g\right)\coloneq\left(\nu_{p_j}(g),\cS_j(\bf/g)\right)_{1 \le j \le n}$ such that
		\begin{equation*}
			\cS_j(\bf/g)
			\coloneq
			\bf/\left(g/p_j^{\nu_{p_j}(g)}\right) \ 
			\bmod
			p_j^{\lambda_j - \nu_{p_j}(g)}.
		\end{equation*}
		By convention, we set $\cS_j(\bf/g) = \Bold{1}$ when $\nu_{p_j}(g) = \lambda_j$.
	\end{defi}
	
	Here we prove, under the hypothesis $N\ge 2FG$, the injectivity of the above encoding.
	
	\begin{prop}
		\label{prop:InjectiveEncodingbad primes}
		Let $\bf/g,\bf' /g'\in\Q^{\ell}$ with $\|\bf\|_{\infty}<F, \ 0<g<G$ such that $\Evi \left(\bf/g\right) = \Evi \left(\bf'/g'\right)$. If we assume that $N \ge 2FG$, 
		the equality $\bf/g = \bf' /g'$ holds.
	\end{prop}
	
	\begin{proof}
		For every $j=1,\ldots,n$ we let $\vr_j \coloneq\nu_{p_j}(g) = \nu_{p_j}(g')$. By hypothesis $\cS_j(\bf/g) = \cS_j(\bf'/g')$, i.e.
		$$
		\bf/\left(g/p_j^{\vr_j}\right)
		=
		\bf' /\left(g'/p_j^{\vr_j}\right) \bmod p_j^{\lambda_j - \vr_j}
		\  \iff \ 
		\bf g'/p_j^{\vr_j}
		=
		\bf' g/p_j^{\vr_j} \bmod p_j^{\lambda_j - \vr_j}.
		$$
		In other words $\bf g' = \bf' g  \
		\bmod p_j^{\lambda_j}$, which implies that $\bf g' = \bf' g \bmod N$. Since by hypothesis $\|\bf g' - \bf' g\|_{\infty}< 2FG\leq N$, we conclude that $\bf g' = \bf' g$ in $\Q^{\ell}$.
	\end{proof}
	
	Under the hypothesis $2FG\leq N$, we can then introduce the \textit{simultaneous rational number code with bad primes} as the set 
	\begin{equation*}
		\SRN^{\infty}_{\ell}(N;F,G)
		\coloneq
		\left\{
		\Evi\left(\frac{\bf}{g}\right)
		:\
		\begin{array}{c}
			\|f\|_{\infty}<F, \quad 0<g<G,\\
			\gcd(f_1,\ldots,f_\ell,g)=1\\
		\end{array}
		\right\}.
	\end{equation*}
	We will refer to it as the SRN code with bad primes.

	Being composed of two parts, codewords $\Evi(\bf/g)$ can be affected by two kinds of errors (valuation and evaluation errors). 
	Here we adapt the hybrid analysis of Section~\ref{sec:Hybrid Decoding}, with the factorization of the error locator $\Lambda = \Li\Lu$ reflecting these two types of errors (see Definition~\ref{def:errorsupportbad primes}).
	
	\begin{defi}
		\label{def:ambientspecebadprimes}
		Let the \textit{ambient space of received words} be the quotient
		\begin{equation*}
			\SL
			\coloneq
			\left(\prod_{j=1}^{n}
			[0,\lambda_j]\times\left(
			\Z/
			p_j^{\lambda_j}\Z
			\right)^\ell\right)/\sim
		\end{equation*}
		where $\sim$ is the equivalence relation for which $(\vr_j,\br_j)_{1 \le j \le
			n}\ \sim\ (\vr'_j,\br'_j)_{1 \le j \le n}$ if and only if  for every
		$j=1,\ldots,n, \ p_j^{\vr'_j}\br_j = p_j^{\vr_j}\br'_j\ \bmod
		p_j^{\lambda_j}$. We say that a representative $(\vr_j,\br_j)_{1 \le j \le
			n}$ is \emph{reduced} if $\gcd(\br_j, p_j^{\vr_j}) = 1$ for every
		$j=1,\ldots,n$. Define $R_i \coloneq \CRT_N(r_{i,1}, \dots, r_{i,n})$ for
		every $i=1,\ldots,\ell$.
	\end{defi}
	
	In what follows we can
	always assume that the received word $(\vr_j,\br_j)_{1 \le j \le n}$ is reduced, thanks to the following proposition:
	\begin{prop}
		Any equivalence class contains a reduced representative.
	\end{prop}
	
	\begin{proof}
		Given any received word $(\vr_j,\br_j)_{1 \le j \le n}$,  for
		every $j=1,\ldots,n$ we let
		$p_j^{\eta_j} \coloneq \gcd(\br_j,p_j^{\vr_j}) $. Then $(\vr_j,\br_j)_{1 \le j \le n}\sim\left(\vr_j -
		\eta_j, \br_j^{\lambda_j}/p_j^{\eta_j} \ \bmod
		p_j^{\lambda_j}\right)_{1 \le j \le n}$, with the representative on the right-hand side
		clearly reduced by the definition of $p_j^{\eta_j}$. 
	\end{proof}
	
	In the ambient space $\SL$ we identify received words which represent the same reduced vector of fractions in the sense that, by definition
	\begin{itemize}
		\item $(\vr_j,\br_j)_{1 \le j \le n}\
		\sim\
		(\vr_j,\br'_j)_{1 \le j \le n}
		\ \Leftrightarrow
		\br_j = \br'_j \ \bmod p_j^{\lambda_j - \vr_j}$.
		\item Given a received valuation $0\leq\vr_j\leq\lambda_j$ then for every $1\leq\delta_j\leq \lambda_j - \vr_j$
		\begin{equation*}
			(\vr_j,\br_j)_{1 \le j \le n}\
			\sim\
			(\vr_j + \delta_j,p_j^{\delta_j}\br_j)_{1 \le j \le n}.
		\end{equation*}
	\end{itemize}
	
	\begin{remark}
		Thanks to the first of the above two points we can map the evaluation of a reduced vector of rationals $\Ev^\infty(\bf/g)$ into the space of received words.
		
	\end{remark}

	\begin{defi}
		\label{def:DistFuncbad primes}
		Given two elements $\bR_1 := (\vr_j,\br_j)_{1 \le j \le n}, \bR_2 :=
		(\vr'_j,\br'_j)_{1 \le j \le n}$ in $\SL$, we
		define the columns $\be_j$ of the relative error matrix $\bE_{\bR_1,\bR_2}$ as
		\begin{equation*}
			\be_j \coloneq
			p_j^{\vr_j} \br'_j - p_j^{\vr'_j} \br_j
			\bmod p_j^{\lambda_j}.
		\end{equation*}   
		We let the relative error and truth locator be
		\begin{equation*}
			\Lambda_{\bR_1,\bR_2} \coloneq \prod_{j=1}^np_j^{\lambda_j - \nu_{p_j}(\be_j)}, 
			\quad
			Y_{\bR_1,\bR_2} \coloneq \prod_{j=1}^n p_j^{\nu_{p_j}(\be_j)}
		\end{equation*}
		respectively, and the relative distance 
		$d\left(\bR_1,\bR_2\right) \coloneq \log\left(\Lambda_{\bR_1,\bR_2} \right)$.
	\end{defi}
	
	\begin{remark}
		\label{rmk:errorsbad primes}
		Unlike the errors considered in Sections~\ref{Sec:SRN}
		and~\ref{sec:Hybrid Decoding}, in this case the usual relation
		$\bR_1 =\bR_2 + \bE$ does not hold.  
		For this reason the error models (see Subsection~\ref{subsect:Errormodelsbad primesRN}) will be defined directly by distributions in the space of received words $\SL$.
	\end{remark}    
	In spite of the above remark, we note the consistency of the error $\be_j$ with the equivalence relation $\sim$, indeed by definition 
	\begin{equation*}
		\be_j = \boldsymbol{0} \bmod p_j^{\lambda_j} \quad 
		\forall j = 1,\ldots,n
		\ \ 
		\Leftrightarrow
		\ \
		(\vr_j,\br_j)_{1\le j \le n} \sim (\vr'_j,\br'_j)_{1\le j \le n}.
	\end{equation*}
	
	Due to the properties of $\sim$, we can partition the set of error positions into valuation and evaluation errors.
	\begin{defi}
		\label{def:errorsupportbad primes}
		Given two evaluations 
		$
		(\vr_j,\br_j)_{1 \le j \le n}$, $(\vr'_j,\br'_j)_{1 \le j \le n} 
		\in \SL
		$ 
		satisfying
		$\gcd(p_j^{\vr_j},\br_j)=1$, 
		we divide the error support
		\begin{equation*}
			\xi=\{j \ | \ p_j^{\vr_j} \br'_j \neq p_j^{\vr'_j} \br_j \bmod p_j^{\lambda_j}\}
			=
			\{j \ | \ (\vr_j,\br_j) \not\sim (\vr'_j,\br'_j)\}
		\end{equation*}
		into the \emph{valuation errors}
		$$\xi_v := \{j \ | \ \vr_j \neq \vr'_j \}$$
		and the remaining \emph{evaluation errors}
		$$\xi_e = \{j \ | \ (\vr_j = \vr'_j) \text{ and } (\br_j \neq \br'_j \bmod
		p_j^{\lambda_j - \vr_j}) \}.$$
	\end{defi}

	We provide an equivalent, yet more practical, representation of the errors.
	
	\begin{remark}
		\label{rem:rewriteErrVectorsbad primes}
		Given a codeword $\left(\nu_{p_j}(g),\cS_j(\bf/g)\right)_{1 \le j \le n}$ (as in
		Definition~\ref{def:MultiprecisionEncoding}) and a received word $(\vr_j,\br_j)_{1 \le j \le n}\in
		\SL$,
		the sequence of error vectors $(\be_j)_{1 \le j \le n}$ is given by
		\begin{equation*}
			\be_j = p_j^{\vr_j} \cS_j(\bf/g) - p_j^{\nu_{p_j}(g)}\br_j \ \bmod p_j^{\lambda_j}.
		\end{equation*}
		Multiplying the above by the invertible element $g/p_j^{\nu_{p_j}(g)}$, we obtain that up to invertible transformations of the error
		sequence components (leaving the distance unchanged), we can equivalently view the sequence of
		error vectors as given by
		\begin{equation*}
			\widetilde{\be}_j \coloneq \frac{g}{p_j^{\nu_{p_j}(g)}} \be_j =p_j^{\vr_j}\bf - g\br_j \ \bmod p_j^{\lambda_j}.
		\end{equation*}
		
	\end{remark}

	\paragraph*{Study of potential errors and received words around a fixed codeword}
	Due to Remark~\ref{rmk:errorsbad primes}, we need to study what kind of errors and received words we can obtain around a fixed vector of fractions $\bf/g$,
	in particular with respect to the distinction between valuation and evaluation errors.
	Regarding the error positions as long as $\xi_e, \xi_v \subset \{1, \dots, n\}$ and $\xi_e \cap \xi_v = \varnothing$ we have no constraints: all valuation (resp. evaluation) error supports $\xi_v$ (resp. $\xi_e$) are attained.
	Once the error positions have been fixed and partitioned as $\xi_v \cup \xi_e$, the valuations of the error vectors need to satisfy $\mu_j = \nu_{p_j}(\be_j) = \lambda_j$ for every position $j$ which is not erroneous, \ie $\forall j \notin \xi_e \cup \xi_v$.
	Let us examine what can happen in the evaluation and valuation error cases respectively:
	\begin{itemize}
		\item If $j\in \xi_e$, we have an evaluation error, thus any received word $\bR$ must satisfy $\vr_j = \nu_{p_j}(g)$, furthermore we must have that the valuation of any error vector $\be_j$ must satisfy $\mu_j = \nu_{p_j}(p_j^{\vr_j}\bf - g\br_j) \ge \nu_{p_j}(g)$ thus, dividing by $p_j^{\vr_j}$, we have that $\cS_j(\bf/g) - r_j$ can be any element of valuation 
		$\mu_j - \nu_{p_j}(g)$.
		\item If $j\in\xi_v$, we have a valuation error, thus for every received word we have either 
		\begin{enumerate}
			\item $\vr_j<\nu_{p_j}(g)$: in this case the valuation of the error vector and the received word must coincide, \ie $\mu_j = \vr_j$, and from the definition of $\widetilde{\be}_j$ we must have that $\widetilde{\be}_j = p_j^{\mu_j}\bf \ \bmod p _j^{\nu_{p_j}(g)}$, regardless of the reduction part $\br_j$. Thus, in this case we do not have any constraints on $\br_j$.
			\item $\vr_j>\nu_{p_j}(g)$: in this case the valuation of the error vector must coincide with the valuation of $g$, \ie $\mu_j = \nu_{p_j}(g)$. Besides this valuation constraint, the error vectors can take any value, as well as the received reductions $\br_j$.
		\end{enumerate}
	\end{itemize}
	
	\paragraph*{Minimal distance}
	Similarly to Lemma~\ref{lm:minDistSRNcode}, we can prove that the Minimal distance of SRN codes with bad primes satisfies the following
	\begin{lemma}
		\label{lm: MinDistSRNbad primes} 
		We have $d\left(\SRN^{\infty}_\ell(N;F,G)\right)>\log_2\left(\frac{N}{2FG}\right).$
	\end{lemma}
	\begin{proof}
		Let $\bC_1 = (\nu_{p_j}(g),\cS_j(\bf/g))_{1\le j \le n}, \bC_2 = (\nu_{p_j}(g'),\cS_j(\bf'/g'))_{1\le j \le n}$ be two distinct codewords. 
		From 
		\begin{equation*}
			\be_j 
			=
			p_j^{\nu_{p_j}(g)}\left(\frac{\bf'}{g'/p_j^{\nu_{p_j}(g')}}\right)
			-
			p_j^{\nu_{p_j}(g')}\left(\frac{\bf}{g/p_j^{\nu_{p_j}(g)}}\right)
			\
			\bmod
			p_j^{\lambda_j},
		\end{equation*}
		we see that
		\begin{equation*}
			\frac{g}{p_j^{\nu_{p_j}(g)}}\frac{g'}{p_j^{\nu_{p_j}(g')}}\be_j =  \bf' g - \bf g' \ \bmod p_j^{\lambda_j}. 
		\end{equation*}
		Using $\Lambda \be_j = 0 \bmod p_j^{\lambda_j}$ for all $j$, we obtain
		$$
		\forall 1 \le j \le n, \quad
		0 = \Lambda \frac{g}{p_j^{\nu_{p_j}(g)}}\frac{g'}{p_j^{\nu_{p_j}(g')}}\be_j = \Lambda (\bf' g - \bf g') \ \bmod p_j^{\lambda_j}.
		$$
		Therefore, $N$ divides $\Lambda (\bf' g - \bf g')$, so $Y = N/\Lambda$ divides $(\bf' g - \bf g')$.
		By the injectivity of the evaluation, $\bC_1 \not\sim \bC_2$ involves $\bf' g - \bf g' \neq \boldsymbol{0}$, 
		which implies that $Y \le \|\bf g' - \bf' g\|_\infty < 2FG$. 
		Hence, for all codewords $\bC_1 \neq \bC_2$, we bound $d(\bC_1,\bC_2) = \log(\Lambda) = \log(N/Y) > \log(N/2FG)$.
	\end{proof}

	\subsection{Key equations}
	The decoding of SRN codes with bad primes, as in Section~\ref{Sec:SRN}, is based on a basis reduction over a lattice describing the solution set to some key equations.
	Thanks to Remark~\ref{rem:rewriteErrVectorsbad primes} and the definition of $\Lambda$, we have that 
	$\Lambda \be_j = 0 \bmod p_j^{\lambda_j}$, and so 
	$0 = \Lambda \widetilde{\be}_j = p_j^{\vr_j} \Lambda \bf - \Lambda g\br_j \bmod p_j^{\lambda_j}$.
	We observe that for every couple of received word $\left(\vr_j,\br_j\right)_{1 \le j \le n}$ and reduced vector of
	fractions $\bf/g$ we have that the equation 
	$
	\CRT_N\left(p_j^{\vr_j}\right) \Lambda f_i
	=
	\Lambda g R_i \ \bmod N
	$
	holds for every $i=1,\ldots,\ell$. By defining the new variables $\varphi \coloneq \Lambda g$,
	$\Bold{\psi} = \Lambda \bf$ we get the \textit{key equations} in presence of bad primes: 
	\begin{equation}
		\label{keyEqRNbad primes}
		\forall i=1,\ldots,\ell, 
		\quad
		\CRT_N\left(p_j^{\vr_j}\right) \psi_i
		=
		\varphi R_i
		\
		\bmod N.
	\end{equation} 
	For some distance parameter $d$, we let the set of solutions be
	\begin{align*}
		\SRd
		\coloneq
		\left\{\left(\varphi,\Bold{\psi}\right)\in\Z^{\ell+1}:
		\begin{array}{c}
			\CRT_N\left(p_j^{\vr_j}\right) \psi_i
			=
			\varphi R_i
			\
			\bmod N,
			\quad\forall i
			\\
			0<\varphi<2^d G,
			\
			\|\Bold{\psi}\|_{\infty}< 2^d F
		\end{array}
		\right\}.
	\end{align*} 
	
	If $\Lambda\leq 2^d$ we see that $v_{\bC}\coloneq(\Lambda g,\Lambda\bf)\in\SRd$.
	
	\paragraph*{Reduced key equations}
	It is possible to give an equivalent description of the solutions in $\SRd$, whose size constraints are smaller.
	Letting $N_{\infty}:=\prod_{j=1}^n p_j^{\vr_j}$ we note that, thanks to Equation~\eqref{keyEqRNbad primes}, $N_{\infty}|\varphi$ since $N_{\infty}|\CRT_N\left(p_j^{\vr_j}\right)$, $N_{\infty}|N$ and by hypothesis $\gcd(N_{\infty},R_i) = 1$ as received words are assumed to be reduced. Thus, we can rewrite Equation~\eqref{keyEqRNbad primes} in the following form, which we call \textit{reduced key equations}
	\begin{equation}
		\label{keyEqRNbad primesReduit}
		\forall i=1,\ldots,\ell, 
		\quad
		\begin{array}{c}
			\psi_i
			=
			\varphi' R_i' \bmod \frac{N}{N_{\infty}} \\
			0<\varphi'< 2^d G/N_{\infty} \text{ and } \|\Bold{\psi}\|_{\infty}< 2^d F
		\end{array}
	\end{equation}
	where $\varphi' \coloneq \varphi/N_{\infty}$ and $R_i' \coloneq R_i \CRT_{N/N_{\infty}}\left(\frac{N_{\infty}}{p_j^{\vr_j}}\right)$.

	\subsection[Decoding SNR codes with bad primes]{Decoding $\SRN^\infty_\ell$ codes}
	\label{subsect:DecodingSRNcodesbad primes}
	In this section we give our decoding algorithm for $\SRN$ codes with bad primes, which is a modification of Algorithm~\ref{algoSRN}. 
	
	As in Section~\ref{Sec:SRN}, the decoding is based on the computation of a short vector $v_s= (\varphi,\bPsi)$ solution of Equations~\eqref{keyEqRNbad primes}. 
	Given the lattice $\L_\infty$ spanned by the
	rows of the matrix
	\begin{equation}
		\label{eq:latticebad primes}
		\L_\infty=\operatorname{Span}\left(\begin{matrix}
			N_{\infty} & R_1' & \cdots & R_\ell' \\
			0 & N /N_{\infty} & \cdots & 0 \\
			\vdots & \vdots & \ddots & \vdots \\
			0 & 0 & \cdots & N /N_{\infty}
		\end{matrix}\right),
	\end{equation}

	we note that all solutions of Equation~\eqref{keyEqRNbad primes} are spanned
	by the rows of $\L_\infty$. Indeed, given a solution $(\varphi,\bPsi)$, thanks to
	Equation~\eqref{keyEqRNbad primesReduit}, we know that 
	$\varphi' = \varphi / N_{\infty}$ is an integer,
	and that
	for $i=1,\ldots,n$ the
	$(i+1)$-th entry $\psi_i - \varphi' R_i'$ of the difference
	$
	(\varphi, \bPsi) - \varphi'(N_{\infty}, R_1',\ldots,R_\ell')
	$ is zero modulo $\frac{N}{N_{\infty}}$.

	Also in this case we introduce a scaling operator $\sigma_{F,G} : \Q^{\ell+1}
	\rightarrow \Q^{\ell+1}$ such that $\sigma_{F,G}((v_0,v_1,\dots,v_\ell))\coloneq(v_0 F,v_1
	G,\dots,v_\ell G)$, to match the size constraints of $\SRd$ with the $\|\cdot\|_\infty$-norm of its elements.
	This scaling will transform $\L_\infty$ into the scaled lattice $\bar{\L}_\infty \coloneq \sigma_{F,G}(\L_\infty)$, and
	our solution set $\SRd$ into
	\begin{equation*}
		\SRd'\coloneq\sigma_{F,G}(\SRd)=\{(\varphi,\psi_1,\ldots,\psi_\ell)\in\bar{\L}_\infty:0<\varphi<2^d FG, \|\bPsi\|_{\infty}<2^d FG\}.
	\end{equation*}

	In the decoding algorithm we compute an element $v_s\in\SRd$ by computing a scaled short vector $\bar{v}_s \coloneq
	\mathcal{ASVP}_{\infty}(\bar{\L}_\infty)$, and unscaling it $v_s \coloneq \sigma_{F,G}^{-1}(\bar{v}_s)$. 
	
	Due to the approximation factor $\beta$ of the sub-routine $\mathcal{ASVP}_\infty$, assuming Constraint~\ref{c_1bad primes}, the solution $v_s$ belongs to a larger set. 
	\begin{constraint}\label{c_1bad primes}Given the received word $(\vr_j,\br_j)_{1\leq j \leq n}$, there exists a code word $\Evi (\bf/g)$ such that the corresponding error locator satisfies $\Lambda \leq 2^d$.
	\end{constraint}
	
	\begin{lemma}
		\label{Lemma1bad primes}
		Assuming Constraint~\ref{c_1bad primes}, we have that $v_s\in \SR \coloneq \SRdb$.
	\end{lemma}
	\begin{proof}
		We know that
		$\|\bar{v}_s\|_\infty\leq\beta\lambda_\infty(\bar{\L}_\infty)\leq\beta\|\sigma_{F,G}(v_C)\|_\infty <
		\beta\Lambda FG \leq \beta 2^d FG$.
		Since we assumed that $(\bar{v}_{s})_0 \ge 0$, we have $\bar{v}_s\in S'_{\bR,2^{d}\beta}$ and
		$v_s\in S_{\bR,2^{d}\beta}$.
	\end{proof}
	We notice that assuming Constraint~\ref{c_1bad primes} we also have $v_{\bC} \in \SR$.
	We are ready to introduce our decoding algorithm for SRN codes with bad primes. 
	\begin{algorithm}
		\caption{$SRN^\infty_\ell$ codes decoder.}
		\label{algoSRNbad primes}
		\SetAlgoLined \SetKw{KwBy}{par} \KwIn{$\SRN^{\infty}_\ell(N;F,G)$, received word $\bR:=(\vr_j,\br_j)_{1\le j\le n}$, distance
			bound $d$} \KwOut{A reduced vector of fractions $\bPsi'/\varphi'$ s.t. $d(\Evi(\bPsi'/\varphi'),\bR) \leq d$ or ``decoding failure''}
		
		\vspace{5pt}

		Let $\bar{\L}_p \coloneq \sigma_{F,G}(\L_\infty)$ be the scaled lattice of $\L_\infty$ defined in
		Equation~\eqref{eq:latticebad primes} \\
		Compute a short vector  $\bar{v}_s \coloneq \mathcal{ASVP}_{\infty} (\bar{\L}_\infty)$ \\
		Unscale the vector: $v_s = (\varphi,\psi_{1},\dots,\psi_{\ell}) \coloneq
		\sigma_{F,G}^{-1}(\bar{v}_s)$ \\
		Let $\eta \coloneq \gcd(\varphi,\psi_{1},\dots,\psi_{\ell})$, $\varphi'\coloneq\varphi/\eta$ and
		$\forall i, \ \psi'_i\coloneq\psi_i/\eta$ \label{step:lambdabad primes}\\
		\If{$\eta \le 2^d$, $|\varphi'| < G$ and $\forall i, \ |\psi_i'|< F$}
		{\textbf{return} $(\psi_1'/\varphi', \dots, \psi_\ell'/\varphi')$}
		\lElse{ \textbf{return} "decoding failure"}
	\end{algorithm}
	\begin{lemma}
		\label{lm:succeedsImpliesCorrectbad primes} 
		If Algorithm~\ref{algoSRNbad primes} returns $\bPsi'/\varphi'$ on input $\bR$ and
		parameter $d$, then $\bPsi'/\varphi'$ is associated to a code word of
		$\SRN^{\infty}_{\ell}(N;F,G)$ close to $\bR$, \ie it is a reduced vector of fractions with
		$\|\bPsi'\|_{\infty}<F$, $0<\varphi'<G$ and
		\mbox{$\dist(\Ev^\infty(\bPsi'/\varphi'),\bR) \le d$}.
	\end{lemma}
	
	\begin{proof}
		The output vector $\bPsi/\varphi$ is associated to a code word of
		$\SRN^{\infty}_{\ell}(N;F,G)$ since the algorithm has verified the size conditions $|\varphi'| < G$,
		$|\psi_i'|< F$ for all $i$.
		Now, we use that $(\varphi,\bPsi)=(\eta \varphi',\eta\bPsi')$ is in the lattice $\L_\infty$, so that $\eta\left(\CRT_N(p_j^{\vr_j}) \bPsi' - \varphi' R_i\right) = 0 \bmod N$ for all
		$i$, which implies that $\nu_{p_j}(\eta)\geq \lambda_j - \mu_j =
		\nu_{p_j}(\Lambda)$ with $\Lambda$ being the error locator between $\Evi (\bPsi/\varphi)$ and the input $\bR$. Thus, $\Lambda|\eta\leq 2^d$, and we can conclude that
		$d(\Ev^\infty(\bPsi'/\varphi'),\bR) = \log \Lambda \le \log \eta \le d$.
	\end{proof}
	
	\subsection{Unique decoding}

	As pointed out in Remark~\ref{rem:failevenifbetais1}, it
	is not because of the approximation factor $\beta$ that Algorithm~\ref{algoSRNbad primes}
	might fail, but because we are decoding with a distance parameter $d>\log(\sqrt{N/2FG})$.
	Thus, at the cost of using an exact SVP solver, \ie\ a
	subroutine $\mathcal{ASVP_\infty}$ returning the shortest vector of
	$\bar{\L}_\infty$, we can assume $\beta = 1$. 
	The drawback of exact SVP solvers is that their complexity is exponential in the dimension of the lattice, nevertheless in our context can be reasonable to employ an exact SVP solver to compute $v_s$, as the dimension $\ell +1$ is fixed and can be assumed to be relatively small. 
	For this reason we prove the unique decoding of SRN codes with bad primes by means of Algorithm~\ref{algoSRNbad primes} with $\beta = 1$ (thus computing an element in $\SRd$) and when the distance parameter $d$ is below unique decoding capacity.
	\begin{prop}
		\label{Unicity}
		If $d(\Evi(\bf/g), (\vr_j,\br_j)_{1 \le j \le n}) \le d \le \log_2\left(\sqrt{N/2FG}\right)$, then $\SRd \subset v_{\bC} \Z$.
	\end{prop}
	\begin{proof}
		By hypothesis $d(\Evi(\bf/g), (\vr_j,\br_j)_{1 \le j \le n}) \le d$, we have $v_{\bC} = (\Lambda g,\Lambda\bf)\in\SRd$.
		Let $(\varphi,\bPsi)\in\SRd$ be another solution of the key equations.
		We have that 
		\begin{equation*}
			\begin{cases}
				p_j^{\vr_j}\Lambda\bf =
				\br_j \Lambda g &\phantom{}\bmod p_j^{\lambda_j}\\
				p_j^{\vr_j}\bPsi\ =
				\br_j \varphi &\phantom{}\bmod p_j^{\lambda_j}
			\end{cases}
		\end{equation*}
		for some 
		$\Lambda\in\Z$ with $\log_2\left(\Lambda\right) \le d
		\le\log_2\left(\sqrt{N/2FG}\right)$. Since the received word $(\vr_j,\br_j)_{1 \le j \le n}$ is assumed to be reduced, \ie $\gcd(p_j^{\vr_j},\br_j) = 1$, from the above we get that 
		$p_j^{\vr_j}|\Lambda g$ and $p_j^{\vr_j}|\varphi$ for every $j=1,\ldots,n$. Thus, when multiplying the first equation by $\varphi$ and the second one by
		$\Lambda g$ we get
		\begin{align*}
			\begin{cases}
				p_j^{\vr_j}\Lambda\varphi \bf =
				\br_j \Lambda g \varphi &\phantom{}\bmod p_j^{\lambda_j + \vr_j}\\
				p_j^{\vr_j}\Lambda g\bPsi  =
				\br_j \Lambda g \varphi &\phantom{}\bmod p_j^{\lambda_j+\vr_j}
			\end{cases}
		\end{align*}
		and subtracting one another, and dividing by $p_j^{\vr_j}$, we obtain 
		$
		\Lambda \left(\varphi\bf - g\bPsi \right) = 0 \bmod N.
		$
		By hypothesis, we have that
		$
		\Lambda \|\varphi\bf - g\bPsi \|_{\infty}
		< 
		2^d (2 F G 2^d )
		\le
		N
		$
		which implies that $\Lambda \left(\varphi\bf - g\bPsi \right) = 0$ thus $\varphi\bf
		= g\bPsi $ in $\Z^\ell$.  
		
		Since $\bf/g$ is a reduced vector of fractions, there exists $a\in\Z$ such that $(\varphi,\bPsi) = a (g,\bf)$.
		Substituting in the key equations for $(\varphi,\bPsi)$, we get
		$
		a (p_j^{\vr_j}\bf - \br_j g) = 0  \bmod p_j^{\lambda_j}
		$.
		However, $\Lambda$ divides $a$ by definition of $\Lambda$, so $(\varphi,\bPsi) \in v_{\bC} \Z$.
	\end{proof}

	\subsection{Hybrid Error Models for Bad Primes}
	\label{subsect:Errormodelsbad primesRN}
	In this subsection we adapt the hybrid error models of the previous section to
	the case with bad primes. Recall that the hybrid error model is composed of both
	fixed errors and random errors. As done in~\cite{guerrini2023simultaneous}, here we consider a hybrid error model where valuation errors are fixed, while evaluation errors are random.
	In previous Sections~\ref{Sec:SRN} and~\ref{sec:Hybrid Decoding}, the error models were defined on the error
	matrices $\bE$, then the theorems applied to received words $\bR$ such that $\bR
	= \bC + \bE$. In this section, as pointed out in Remark~\ref{rmk:errorsbad primes}, we have a more complicated relation between
	$\bR$, $\bC$ and $\bE$. So we are going to define the error model directly on
	$\bR$.

	Our error model needs to fix the following parameters:
	\begin{itemize}
		\item a reduced vector of rationals $\bf/g\in\Q^{\ell}$ such that
		$
		\|\bf\|_{\infty}<F,$ $ 0<g<G
		$,
		
		\item valuation $\xi_v$ and evaluation $\xi_e$ error supports such that
		$\xi_e, \xi_v \subset \{1, \dots, n\}$ and $\xi_e \cap \xi_v = \varnothing$,
		
		\item error valuations $(\mu_j)_{1 \le j \le n}$ such that 
		\begin{itemize}
			\item $\mu_j = \lambda_j$ for $j \notin \xi_e \cup \xi_v$,
			\item $\mu_j \ge \nu_{p_j}(g)$ and $\mu_j < \lambda_j$ for $j \in \xi_e$,
			\item $\mu_j \le \nu_{p_j}(g)$ and $\mu_j < \lambda_j$ for $j \in \xi_v$,
		\end{itemize}
		
		\item a partial received word $\bfR_j = (\vr_j, \br_j)$ for all $j \in
		\xi_v$ such that 
		\begin{itemize}
			\item $\vr_j = \mu_j$ when $\mu_j < \nu_{p_j}(g)$,
			\item $\vr_j > \nu_{p_j}(g)$ when $\mu_j = \nu_{p_j}(g)$.
		\end{itemize}
	\end{itemize}
	Denote $\Le := \prod_{j \in \xi_e} p_j^{\lambda_j - \mu_j}, \Lv := \prod_{j \in
		\xi_v} p_j^{\lambda_j - \mu_j}$ and $\Lambda = \Le \Lv$. Remark that $\Le, \Lv,
	\Lambda$ contain all the information of $\xi_v, \xi_e$ and $\mu_j$ since $\xi_v
	= \cP(\Lv)$, $\xi_e = \cP(\Le)$ and $\mu_j = \lambda_j - \nu_{p_j}(\Lambda)$.
	
	We are ready to define our error models. The random received words $\bR = (\vr_j, \br_j)_j$ are uniformly
	distributed in the following set $\HYBp$
	\begin{enumerate}
		\item $\bR_j = \Evi(\bf/g)_j \text{ for all } j \text{ such that } p_j\not\in\cP(\Lambda)$,
		\item $\bR_j = \bfR_j \text{ for all } j\text{ such that } p_j\in\cP(\Lv)$,
		\item $\bR_j = (\nu_{p_j}(g), \br_j)$ with $\nu_{p_j}\left(\br_j -
		\cS_j(\bf/g)\right) = \mu_j - \nu_{p_j}(g)$ for all $j$ such that $p_j\in\cP(\Le)$.
	\end{enumerate}
	
	As before, we will determine the distribution of the error matrices
	$\bE_{\bR,\Evi(\bf/g)}$ when $\bf/g$ is fixed and $\bR$ is random.
	
	For $i\in\{1,\dots,\ell\}$, we still denote $E_i \in \Z/N\Z$ the CRT interpolant of
	the $i$-th row of $\bE$, and we obtain that $Y | E_i$ for every index $i=
	1,\ldots,\ell$ as in Subsection~\ref{subsect:ErrModelRatNumb}. We define the
	modular integers $E_i'\coloneq E_i/Y \in \Z/\Lambda\Z$, which verify
	$\gcd(E_1',\ldots,E_{\ell}',\Lambda) = 1$.
	
	Because of our hybrid error model where the randomness only appears on the
	columns $j \in\cP(\Le)$, we need to study the random vector $(E_1' \bmod \Le,\ldots,E_{\ell}' \bmod \Le)$.
	\begin{lemma}
		\label{lm:randombad primesuniform}
		If $\bR$ is uniformly distributed in $\HYBp$, then 
		the random vector $(E_1' \bmod \Le,\ldots,E_\ell' \bmod \Le)$ is uniformly
		distributed in the sample space $\Omega_{\Lambda_e}$.
	\end{lemma}
	\begin{proof}
		For the duration of this proof, we will only consider indices $j$ such that
		$p_j\in\cP(\Le)$. Recall that
		$
		\be_{j} = p_j^{\nu_{p_j}(g)} (\br_{j} - \cS_j(\bf/g)) \bmod p_j^{\lambda_j}
		$
		for all those particular $j$. Since $Y = p_j^{\mu_j} \bmod p_j^{\lambda_j}$,
		we get that $E'_i = E_i/Y = e_{i,j}/Y \bmod p_j^{\lambda_j - \mu_j}$ and
		$$
		\be_{j}/Y = (\br_{j} - \cS_j(\bf/g))/p_j^{\mu_j - \nu_{p_j}(g)} \bmod p_j^{\lambda_j - \mu_j}.
		$$
		Therefore, by definition of $\HYBp$, the vector $\be_{j}/Y \in
		(\Z/p_j^{\lambda_j - \mu_j}\Z)^\ell$ is random of valuation $0$. As a
		consequence, we obtain that $(E_1' \bmod \Le,\ldots,E_\ell' \bmod \Le)$ is random
		among the vectors of $(\Z/\Le\Z)^\ell$ such that
		$\gcd(E_1',\ldots,E_{\ell}',\Le) = 1$.
	\end{proof}
	
	\paragraph*{Second error model}
	Similarly, we need to fix a reduced vector of rationals $\bf/g\in\Q^{\ell}$,
	valuation $\xi_v$ and evaluation $\xi_{m,e}$ error supports, error valuations
	$(\mu_j)_{1 \le j \le n}$ and a partial received word $\bfR_j = (\vr_j, \br_j)$
	for all $j \in \xi_v$. All these parameters must satisfy the same conditions as
	the first error model. 
	
	The set $\xi_{m,e}$ is now called the maximal error
	support because actual errors could result in an evaluation error support
	$\xi_{e} \subset \xi_{m,e}$.
	
	Denote $\Lambda_{m,e} := \prod_{j \in \xi_{m,e}} p_j^{\lambda_j - \mu_j}$,
	$\Lv := \prod_{j \in \xi_v} p_j^{\lambda_j - \mu_j}$ and $\Lambda_m = \Lambda_{m,e} \Lv$.

	In the second error model, the random received words $\bR = (\vr_j, \br_j)_j$
	are uniformly distributed in the following set $\HYBpd$
	\begin{enumerate}
		\item $\bR_j = \Evi(\bf/g)_j \text{ for all } j \text{ such that } p_j\not\in\cP(\Lambda_m)$,
		\item $\bR_j = \bfR_j \text{ for all } j\text{ such that } p_j\in\cP(\Lv)$,
		\item $\bR_j = (\nu_{p_j}(g), \br_j)$ with $\nu_{p_j}\left(\br_j -
		\cS_j(\bf/g)\right) \ge \mu_j - \nu_{p_j}(g)$ for all $j$ such that $p_j\in\cP(\Lambda_{m,e})$.
	\end{enumerate}
	Notice that for a given received word in the set $\HYBpd$, the associated
	error locator has the form $\Lambda = \Le\Lv$ for some divisor $\Le|\Lme$.

	\subsection{Our results on bad primes}

	We are ready to state our results regarding the failure probability of the decoding algorithm in
	presence of bad primes. We let $\dme$ be the maximal distance on the evaluation errors 
	\begin{equation}
		\label{def:Dmaxbad primes}
		\dme \coloneq 
		\frac{\ell}{\ell+1} \left[\log(N/2FG) - \log(3\beta) - 2 \hdv\right]
	\end{equation}
	\begin{theorem}
		\label{thm:main1bad primes}
		Decoding Algorithm~\ref{algoSRNbad primes} on input 
		\begin{enumerate}
			\item distance parameter $d = \hdv
			+ \hde$ for $\hdv\leq \log\left(\sqrt{N/(6FG\beta)}\right)$
			and $\hde\leq \dme$,
			\item a random received word $\bR = (\vr_j,\br_j)_{1 \le j \le n}$
			uniformly distributed in $\HYBp$, for some reduced vector of fractions
			$\bf/g\in\Q^{\ell}$ with $\|\bf\|_{\infty} < F, \ 0<g<G$, and $\log
			(\Lv) \leq \hdv$ and $\log(\Le)\leq \hde$,
		\end{enumerate}
		outputs the center vector $\bf/g$ of the distribution
		with a probability of failure 
		\begin{equation*}
			\P_{\fail}
			\leq 
			2^{-(\ell+1)(\dme - \hde)}\prod_{p\in\cP(\Le)}\left(\frac{1 - 1/p^{\ell + \nu_p(\Lambda_{e)}}}{1 - 1/p^{\ell}}\right).
		\end{equation*} 
	\end{theorem}
	\begin{theorem}
		\label{thm:main2bad primes}
		Decoding Algorithm~\ref{algoSRNbad primes} on input 
		\begin{enumerate}
			\item distance parameter $d = \hdv
			+ \hde$ for $\hdv\leq \log\left(\sqrt{N/(6FG\beta)}\right)$
			and $\hde\leq \dme$,
			\item a random received word $\bR = (\vr_j,\br_j)_{1 \le j \le n}$
			uniformly distributed in $\HYBpd$, for some reduced vector of fractions
			$\bf/g\in\Q^{\ell}$ with $\|\bf\|_{\infty} < F, \ 0<g<G$, and $\log
			(\Lv) \leq \hdv$ and $\log(\Lme)\leq \hde$,
		\end{enumerate}
		outputs the center vector $\bf/g$ of the distribution
		with a probability of failure 
		\begin{equation*}
			\P_{\fail}
			\leq 
			2^{-(\ell+1)(\dme - \hde)}
			\prod_{p\in\cP(\Lme)}
			\left(\frac{1 - 1/p^{\ell + \nu_p(\Lme)}}{1 - 1/p^{\ell + 1}}\right).
		\end{equation*} 
	\end{theorem}
	
	\begin{remark}
		The results presented here have a polynomial counterpart in the context
		of rational function codes with multiplicities and poles as studied
		in~\cite{guerrini2023simultaneous}.
		
		We remark that the results given in this paper provide several
		improvements on the state of the art of the polynomial counterpart
		(see~\cite[Theorem 3.4]{guerrini2023simultaneous}). For instance, the
		failure probability bound decreases exponentially when the actual error
		distance is less than the maximal error distance in this paper, whereas the failure probability bound in~\cite{guerrini2023simultaneous} is a linear function of the distance parameter. Furthermore, our bound removes the technical dependency of the multiplicity balancing, making the results independent of how the multiplicities are distributed.
		
		In this work, we establish our results only in the setting of rational
		numbers. Following similar lines of reasoning, we have also obtained
		analogous theorems in the case of rational functions; however, to keep
		the paper concise, enhance readability, and focus on the more original
		contributions, we have opted not to include these statements. 
	\end{remark}

	\subsection{Decoding failure probability with respect to the first error model} 
	We let 
	\begin{equation*}
		\SEinf\coloneq
		\left\{
		\omega\in\Z/\Le\Z : \forall i, \ \omega \widetilde{E}'_i \in \Z_{\Le,B\Lambda}
		\right\}
	\end{equation*}
	with $B\coloneq2^{d}\beta\frac{2FG}{N} = 2^{\hde+\hdv}\beta\frac{2FG}{N} $ and
	$\widetilde{E}_i \coloneq \CRT_N\left(g/p_j^{\nu_{p_j}(g)}\right) E_i \ \bmod N$. 
	We can now prove the version of Lemma~\ref{BaseLemmaRN} with bad primes.
	\begin{constraint}
		\label{cst: Parambad primesVers}
		The parameters of Algorithm~\ref{algoSRN} satisfy $2^{\hdv}B< 1$.
	\end{constraint}
	\begin{lemma}
		\label{lm:SE_bad primes}
		If Constraint~\ref{cst: Parambad primesVers} is satisfied then $\SEinf = \{0\} \Rightarrow
		S_{\bR}\subseteq v_{\bC}\Z$.
	\end{lemma}
	\begin{proof}
		Let $(\varphi,\psi_1,\ldots,\psi_\ell)\in S_{\bR} = \SRdb$. From~\eqref{keyEqRNbad primes} we know that $\prod_{j=1}^np_j^{\vr_j}|\varphi$ and that for every $i,j$ there exists $h_{i,j}\in\Z$ such that $\varphi r_{i,j}=p_j^{\vr_j}\psi_i + h_{i,j} p_j^{\lambda_j}$. Furthermore,
		\begin{equation}
			\label{eq:intermediatefromLemmaSEbad primes}
			\varphi \Lv \widetilde{e}_{i,j}
			= 
			p_j^{\vr_j}\Lv\left(\varphi f_i - g\psi_i\right) - \Lv g h_{i,j} p_j^{\lambda_j} \ \bmod p_j^{\lambda_j + \vr_j}.    
		\end{equation}
		From
		\begin{equation*}
			\nu_{p_j}(\Lv g)=
			\begin{cases}
				\lambda_j - \min\{\vr_j,\nu_{p_j}(g)\} + \nu_{p_j}(g) & \text{ if }\vr_j\ne\nu_{p_j}(g)\\
				\nu_{p_j}(g)& \text{ if }\vr_j=\nu_{p_j}(g)
			\end{cases},
		\end{equation*}
		as $\lambda_j\geq \vr_j $, we conclude that $\nu_{p_j}(\Lv g)\geq \vr_j$ for every $j=1,\ldots,n$. Taking the CRT interpolant modulo $N$ on both sides of~\eqref{eq:intermediatefromLemmaSEbad primes} after dividing by $p_j^{\vr_j}$, we conclude that
		$$
		\CRT_N(\varphi/p_j^{\vr_j})\Lv \widetilde{E}_{i} 
		=
		\Lv
		(
		\varphi f_i - g \psi_i
		)
		\ \bmod N
		$$
		with $\widetilde{E}_i \coloneq \CRT_N\left(g/p_j^{\nu_{p_j}(g)}\right) E_i \ \bmod N$.
		The integer $Y\Lv$ divides both $\Lv \widetilde{E}_{i}$ and $N$, so it divides  
		$\Lv ( \varphi f_i - g \psi_i )$.
		Dividing by $Y\Lv$, we obtain 
		\begin{equation*}
			\CRT_{\Le}\left(\frac{\varphi}{p_j^{\vr_j}}\right)\widetilde{E}'_{i} 
			=
			\frac{\varphi f_i - g \psi_i}{Y}
			\ \bmod \Le,
		\end{equation*}
		with $\widetilde{E}'_i \coloneq \CRT_{\Le}\left(g/p_j^{\nu_{p_j}(g)}\right) E'_i \ \bmod
		\Le$. Thus, $\omega \coloneq \CRT_{\Le}\left(\varphi/p_j^{\vr_j}\right)\in      
		\SEinf$ and, thanks to the hypothesis $\SEinf = \{0\}$, $\left(\varphi
		f_i - g \psi_i\right)/Y = 0 \ \bmod\Le$. Thanks to Constraint~\ref{cst:
			Parambad primesVers} and since $\Lv\leq 2^{\hdv}$, we have
		$
		\frac{2FG}{N}2^d \beta \Lv < 1
		$
		which implies 
		$
		\frac{|g\psi_i-f_i\varphi|}{Y} \leq \frac{2FG}{N} 2^d \beta \Lambda < \Le.
		$
		As a result, $g\psi_i = f_i\varphi$ for all $i=1,\ldots,\ell$.
		Since $\gcd(f_1,\ldots,f_\ell, g)=1$,  we must have that $g|\varphi$, i.e. $\varphi = sg$ for
		some $s\in\Z$ and, from the above conclusion, as well that $\Bold{\psi} = s\bf$. Let us
		note 
		\begin{equation*}
			s\widetilde{\be}_j =
			p_j^{\vr_j}\Bold{\psi} - \varphi\br_j = \Bold{0} \bmod p_j^{\lambda_j}.
		\end{equation*} 
		As $\nu_{p_j}\left(\widetilde{\be}_j\right) = \nu_{p_j}\left(\be_j\right) = \lambda_j -
		\Val_j(\Lambda)$, we obtain $\nu_{j}(s)\ge \lambda_j - (\lambda_j - \Val_j(\Lambda))$
		for every $j$, i.e. $\Lambda$ divides $s$.
	\end{proof}
	
	\begin{remark}
		Along the same lines of Remark~\ref{rmk:unicity} relative to the analysis of Section~\ref{Sec:SRN}, also in this context we see that, at the cost assuming of $\beta = 1$, our technique yields the unique decoding when the distance parameter $d$ of Algorithm~\ref{algoSRNbad primes} is below unique decoding capacity. Indeed, when $d<\log(\sqrt{N/(2FG)})$, we must have that $B\Lambda< \beta$ since
		$\Lambda\le 2^d$.
		Under such circumstance we therefore
		have $\Z_{\Le, B\Lambda} = \Z_{\Le, 0} = \{0\}$. Thus, also in this case, estimating the
		failure probability of Algorithm~\ref{algoSRNbad primes} by studying $\P(\SEinf\ne\{0\})$
		yields the expected unique decoding result when $d<\log(\sqrt{N/(2FG)})$.
	\end{remark}
	
	\begin{proof}[Proof of Theorem~\ref{thm:main1bad primes}]
		Since Constraint~\ref{c_1bad primes} is verified for all received words in
		our random distribution, Lemma~\ref{Lemma1bad primes} and an adaptation of
		Lemma~\ref{lm:algofailurecondition} shows that 
		$\P_{\fail} \le \P(S_{\bR}\not\subseteq v_{\bC}\Z)$.
		
		We can prove that our choice of parameters satisfy Constraint~\ref{cst:
			Parambad primesVers} in the same fashion as the proof of
		Theorem~\ref{thm:main1hyb}. So we can apply Lemma~\ref{lm:SE_bad primes} to
		obtain $\P(S_{\bR}\not\subseteq v_{\bC}\Z) \le \P\left(\SEinf \ne
		\{0\}\right)$.
		
		As in Equation~\eqref{fstUpperBoundProof}, we have 
		$
		\P\left(\SEinf \ne \{0\}\right)
		\leq
		\sum_{\omega=1}^{\Le - 1}
		\P\left(
		\forall i, \  \omega \widetilde{E}_i' \in \Z_{\Le,B \Lambda}
		\right)
		$.
		Since $\widetilde{E}_i = \CRT_N\left(g/p_j^{\nu_{p_j}(g)}\right) E_i \ \bmod N$, and since
		$\CRT_N\left(g/p_j^{\nu_{p_j}(g)}\right)$ is an invertible element of $\Z/N\Z$, we have that for
		every $1\leq\omega\leq\Le - 1$,\linebreak
		$
		\P(
		\forall i, \  \omega \widetilde{E}_i' \in \Z_{\Le,B \Lambda}
		) 
		=
		\P(
		\forall i, \  \omega E_i' \in \Z_{\Le,B \Lambda}
		).
		$
		Now, since we know the distribution of $(E_i')_{1 \le i \le n}$ thanks to
		Lemma~\ref{lm:randombad primesuniform}, we use
		Lemma~\ref{lm:boundfailureprobabilityhybrid} with $\Li$ and $\hdu$ being
		replaced by $\Le$ and $\hdv$ to get
		\begin{equation*}
			\sum_{\omega=1}^{\Le - 1}
			\P\left(
			\forall i, \  \omega E_i' \in \Z_{\Le,B \Lambda}
			\right) 
			\leq
			\left(
			3B 2^{\hdv}\right)^\ell \Le
			\prod_{p\in\cP(\Le)}\left(\frac{1 - 1 / p^{\ell + \nu_p(\Le)}}{1 -1 / p^{\ell}}
			\right).
		\end{equation*}
		Since $\Le\leq 2^{\hde}$, we have
		$
		\left(3B 2^{\hdv}\right)^\ell \Le 
		\le 
		\left(3B\right)^\ell 2^{\ell \hdv} 2^{\hde}
		=
		2^{-(\ell+1)(\dme - \hde)}
		$, we have proven Theorem~\ref{thm:main1bad primes}.
	\end{proof}
	
	\subsection{Decoding failure probability with respect to the second error model}
	\label{subsubsect: proofbad primesRN-D2}
	
	We will denote $\P_{\HYBpd}$ (resp. $\P_{\HYBp}$) the probability function under
	the second (resp. first) error model specified by a given factorization $\Lambda_{m,e}\Lv$ of
	the error locator, and by a partial received word $(\bfR_j)_{j \in \cP(\Lv)}$.
	
	\begin{proof}[Proof of Theorem~\ref{thm:main2bad primes}]
		As done in the proof of Theorem~\ref{thm:main2hyb}, letting $\cF$ be the event of decoding failure.
		We will denote $\P_{\HYBpd}$ (resp. $\P_{\HYBp}$) the probability function under
		the second (resp. first) error model specified by a given factorization $\Lambda_{m,e}\Lv$ of
		the error locator, and by a partial received word $(\bfR_j)_{j \in \cP(\Lv)}$.
		Using the law of total probability, we have that
		$\P_{\HYBpd}( \cF )$ can be decomposed as the sum
		\begin{equation*}
			\sum_{\Le|\Lme}
			\P_{\HYBpd}(\cF \ |\ \Lambda_{\bE} = \Le\Lv)
			\ 
			\P_{\HYBpd}( \Lambda_{\bE} = \Le\Lv),
		\end{equation*}
		where 
		\begin{equation*}
			\P_{\HYBpd}(\cF \ |\ \Lambda_{\bE} = \Le\Lv)
			=
			\P_{\HYBp}\left(\cF\right)
		\end{equation*}
		is upper bounded by
		\begin{equation*}
			\P_{\HYBp}(\cF)
			\leq
			\left(
			3B 2^{\hdv}\right)^\ell \Le
			\prod_{p\in\cP(\Le)}\left(\frac{1 - 1 / p^{\ell + \nu_p(\Le)}}{1 -1 / p^{\ell}}\right)
		\end{equation*}
		Whereas
		\begin{align*}
			\P_{\HYBpd}(\cF \ |\ \Lambda_{\bE} = \Le\Lv)
			&=
			\frac{\prod_{p\in\cP(\Le)}(p^{\ell} -1) p^{\ell(\nu_p(\Le) - 1)}}{\prod_{p\in\cP(\Lme)}p^{\ell \nu_p(\Lme)}}\\
			&=
			\left(\frac{\Le}{\Lme}\right)^{\ell}
			\prod_{p\in\cP(\Le)}\left(1 - \frac{1}{p^{\ell}}\right).
		\end{align*}
		Plugging the above in the decomposition of $\P_{\HYBpd}( \cF )$ and following the
		proof of Theorem~\ref{thm:main2hyb} with $\Lmi,\Li,\Lu$ being replaced by
		$\Lme,\Le,\Lv$ respectively, we conclude the proof of Theorem~\ref{thm:main2bad primes}.
	\end{proof}
	
	\bibliographystyle{alpha}
	\bibliography{Bib.bib}
	
\end{document}